\documentclass[journal,twocolumn]{IEEEtran}
\makeatletter

\let\proof\@undefined
\let\endproof\@undefined
\makeatother

\usepackage{epsfig,amsmath,amssymb,epsf,cite,subfigure,amsthm,algorithm,algorithmic,amsbsy}
\newcommand{\argsup}{\mathop{\mathrm{argsup}}}

\newcommand{\argmax}{\mathop{\mathrm{argmax}}}
\newcommand{\var}{\mathop{Var}}
\newtheorem{theorem}{Theorem}
\newtheorem{lemma}{Lemma}

\newtheorem{conjecture}[theorem]{Conjecture}

 \def\cN{{\mathcal{N}}}

\def\argmax{\mathop{\mathrm{argmax}}}

\def\b0{{\pmb{0}}}

 \def\bb{{\mathbf{b}}} \def\bc{{\mathbf{c}}} 
 \def\bff{{\mathbf{f}}} \def\bg{{\mathbf{g}}} 
   
 \def\bn{{\mathbf{n}}}  
\def\bq{{\mathbf{q}}}   
   \def\bx{{\mathbf{x}}}
\def\by{{\mathbf{y}}} \def\bz{{\mathbf{z}}} \def\bth{\boldsymbol{\theta}}

\def\bA{{\mathbf{A}}}  \def\bC{{\mathbf{C}}} 
 \def\bF{{\mathbf{F}}}  
\def\bI{{\mathbf{I}}}

\begin{document}
\title{Concatenated Coding for the AWGN Channel with Noisy Feedback}
\date{}

\author{Zachary Chance, \emph{Student Member}, \emph{IEEE}, and David J. Love, \emph{Senior Member}, \emph{IEEE}\\
\thanks{
This work was presented in part at the Asilomar Conference on Signal, Systems, and Computers in November 2009 and at the International Conference on
Acoustics, Speech, and Signal Processing in March 2010.

This work was supported in part by the National Science Foundation under grant CCF0513916.  Zachary Chance, a Ph.D.
student, and David J. Love, an associate professor, are affiliated
with the School of Electrical and Computer Engineering at Purdue University, West Lafayette, IN, 47907.  Email: zchance@purdue.edu, djlove@ecn.purdue.edu.
}
}

\maketitle

\begin{abstract}
The use of open-loop coding can be easily extended to a closed-loop concatenated code if the transmitter has access to feedback.  This can be done by introducing a feedback transmission scheme as an inner code.  In this paper, this process is investigated for the case when a linear feedback scheme is implemented as an inner code and, in particular, over an additive white Gaussian noise (AWGN) channel with noisy feedback.  To begin, we look to derive an optimal linear feedback scheme by optimizing over the received signal-to-noise ratio.  From this optimization, an asymptotically optimal linear feedback scheme is produced and compared to other well-known schemes.  Then, the linear feedback scheme is implemented as an inner code to a concatenated code over the AWGN channel with noisy feedback.  This code shows improvements not only in error exponent bounds, but also in bit-error-rate and frame-error-rate.  It is also shown that if the concatenated code has total blocklength $L$ and the inner code has blocklength, $N$, the inner code blocklength should scale as $N = O\left(\frac{C}{R}\right)$, where $C$ is the capacity of the channel and $R$ is the rate of the concatenated code.  Simulations with low density parity check (LDPC) and turbo codes are provided to display practical applications and their error rate benefits.
\end{abstract}

\begin{keywords}
additive Gaussian noise channels, concatenated coding, linear feedback, noisy feedback, Schalkwijk-Kailath coding scheme
\end{keywords}

\section{Introduction}
\PARstart{T}{H}{E} field of open-loop error-correction coding has been rich with innovations over the last 10-20 years with implementation of codes like turbo codes and low density parity check (LDPC) codes.  These codes have proven that open-loop methods can be very powerful.  However, an important question to be asked is: ``Can we do better with closed-loop coding?''  In this paper, we investigate the use of closed-loop \emph{concatenated coding} (see Fig. \ref{concode}) over an additive white Gaussian noise (AWGN) channel with noisy feedback.  The benefits of this have already been shown for a noiseless feedback channel as in \cite{Cain}.

By definition, concatenated coding consists of two codes: an inner code and an outer code.  As for the outer code, we will assume it is any general forward error-correction code as to make this method applicable to any open-loop technique.  Furthermore, since we are interested in closed-loop coding, the inner code will be a feedback transmission scheme; this, however, creates the intermediate goal of designing a \emph{good} feedback scheme.  In this, we narrow our focus to the class of linear feedback transmission schemes - meaning that each transmission is a linear function of feedback side-information and the message to be sent.  With perfect feedback, this class is known to have low complexity and high reliability \cite{Schal1, Schal2}.  Therefore, we will try to exploit these advantages even with a noisy feedback channel.

\begin{figure*}
\centering
  \includegraphics[scale = 0.9]{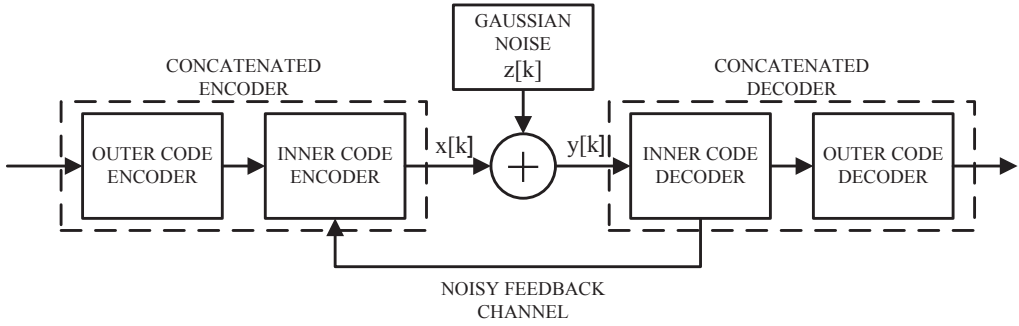}\\
  \caption{Closed-loop concatenated coding system.}\label{concode}
\end{figure*}

The search for the best linear feedback coding scheme for AWGN channels has a long history, dating back to 1956 with a paper by Elias \cite{Elias1}.  However, most early work was done in the 1960's with papers like \cite{Butman, Omura, Elias2, Turin1, Turin2, Horstein, Ferg, Kashyap}\nocite{Cover}.  In 1966, Schalkwijk and Kailath developed a specific linear coding technique that utilizes a noiseless feedback channel\cite{Schal1, Schal2}.  The coding scheme was based off of a zero-finding algorithm called the Robbins-Monro procedure which sequentially estimates the zero of a function given noisy observations.  Because of its low complexity, much work has been done extending and evaluating the performance of the Schalkwijk-Kailath (S-K) scheme in different circumstances.  The performance was examined when there is bounded noise on the feedback channel in \cite{Nuno}. In \cite{Kramer,Wyner}, the system was observed under a peak energy constraint.  A generalization of the coding scheme for first-order autoregressive noise processes on the forward channel was derived in \cite{Butman} while the problem was also looked at in \cite{Wolf1, Tier}.  The use of the coding technique was extended to applications in stochastic controls in \cite{Omura}.  It was also extended for use in stochastically degraded broadcast channels in \cite{Oz} and two-user Gaussian multiple access channels in \cite{Oz2}.  The scheme was used in \cite{Kim} for a derivation of feedback capacity for first-order moving average channels and, in general, for channels with stationary Gaussian forward noise processes in \cite{Kim4}.  In \cite{Gallager}, it was reformulated using a previous result in \cite{Elias1} and then altered for specific use with PAM signaling.  Variations on it were created by using stochastic approximation in \cite{Kumar}.  The S-K scheme was also used in a derivation of an error exponent for AWGN channels with partial feedback in \cite{Agar}.  This brief overview, of course, is not exhaustive as much more literature can be found on the subject.  In fact, due to the notable popularity of the S-K scheme, we will implement it for performance comparisons.

Recently, the area of general feedback communication schemes has been also studied as in \cite{shay,cole}.  These use a technique called \textit{Posterior Matching} in which information at the receiver is refined using the a-posteriori density function which is matched to the input distribution function.  Such techniques have also proven to be capacity-achieving and, in fact, a generalization of the S-K scheme.  However, these schemes along with the S-K scheme all rely on the presence of a noiseless feedback channel to achieve non-zero rate.  If noise is present, all of these schemes have only an achievable rate of zero.  However, coding with the presence of a noisy feedback channel with variable length techniques has been investigated in \cite{Sahai1,Sahai2}.

In this paper, we do the following:
\begin{itemize}
\item Give the basic framework of concatenated coding and design a linear feedback scheme for use as an inner code by:
\begin{itemize}
    \item Using a matrix formulation for feedback encoding, we formulate the maximum SNR optimization problem.  The formulation consists of a combining vector and noise encoding matrix.  It shares many similarities to the method employed by \cite{Butman}.  In addition, an upper bound on SNR for all noisy feedback schemes over the AWGN channel is derived and shown to be tighter than the bound previously made by \cite{Butman}.
    \item Using SNR as the cost function of interest, we solve for (i) the SNR-maximizing linear receiver given a fixed linear transmit encoding scheme and (ii) the SNR-maximizing linear transmitter given a fixed linear receiver.
    \item Using insights from the numerical optimization, we derive what we believe to be the optimal linear processing set-up.  The performance of the proposed scheme approaches the linear processing SNR upper bound as the blocklength grows large.
\end{itemize}
\item Using the proposed linear feedback scheme, we then implement a concatenated code over the AWGN channel with a general error-correction code as an outer code.  The error exponent for the concatenated scheme is derived in terms of the error exponent for the outer code.
\item Upper and lower bounds on the feedback error exponent are then derived using this setup.  These bounds are then used to illustrate the effect of using the proposed linear feedback scheme as an inner code.  An approximate trade-off between inner code blocklength and total code blocklength is also derived.
\item Simulations are run to show advantages in bit-error-rate (BER) and frame-error-rate (FER) when the outer code is either a turbo code or LDPC code.
\end{itemize}

The paper is organized in the following manner.  The overall system and the framework for a general closed-loop concatenated coding scheme are introduced in Section II.  In Section III, the concept of linear feedback coding is introduced to develop an appropriate inner code for the overall concatenated scheme.  Two methods of optimization for a general linear coding scheme are also briefly introduced.  Using these optimization methods, a linear feedback scheme is proposed in Section IV.  This section also consists of analyzing the asymptotic performance of our scheme, along with deriving alternate proofs of results from related papers.  In Section V, the proposed linear feedback scheme is compared against the S-K scheme to illustrate gains in performance.  Section VI introduces the concatenated coding scheme and its error rate analyses.  Simulations are then given in Section VII to demonstrate practical concatenated code performance with turbo codes and LDPC codes.

\section{General Closed-Loop Concatenated Coding}
In this section, we formulate the general framework for a closed-loop concatenated coding scheme; to begin, we look specifically at the AWGN channel.  At each channel use $k = 1,2,\ldots,L$, the transmitted signal, $x[k] \in \mathbb{R}$, is sent across the channel. Likewise, the receiver obtains
\begin{equation}
y[k] = x[k] + z[k],
\label{channel}
\end{equation}
\noindent where, for our purposes, we will assume $\{z[k]\} \in \mathbb{R}$ are i.i.d. such that $z[k]\sim\cN(0,1)$.  Also, to remain practical, we can impose an average transmit power constraint, $\rho$, such that
\begin{equation}
E[\bx^T\bx] \leq L\rho,
\label{allpower}
\end{equation}
\noindent where $\bx = \left[x[1],x[2],\ldots,x[L]\right]^T.$

Consider sending a length $K$ open-loop code across the AWGN channel with noisy feedback.  The transmission of each component of the open-loop codeword, $\bc \in \mathbb{R}^{K}$, will be encoded using an inner code of blocklength $N$ that has access to noisy feedback.  Thus, the total concatenated codeword and accordingly, the transmit vector, $\bx$, has length $L = KN$.  Note that the open-loop codeword is composed of entries that lie on the real line; this implies that if the outer code is binary, a modulation operation is implicit.  Now, if we write the components of the inner code as $s_{i}(c_{j})$ (the $i$-th inner code component used to encode the $j$-th outer code component), then $\bx$ can alternatively be written as
\[
\bx = \left[s_{1}(c_{1}),s_{2}(c_{1}),\ldots,s_{N}(c_{1}),s_{1}(c_{2}),\ldots,s_{N}(c_{K})\right].
\]
\noindent This encoding process now can be grouped by the concatenated encoder (or \textit{superencoder}) in Fig. \ref{concode}.  At this point, we also bound the average power of the outer codeword as
\begin{equation}
E[\bc^T\bc] \leq K\rho.\label{cpower}
\end{equation}
The codeword power constraint (\ref{cpower}) will be useful when we analyze the performance of the concatenated scheme in Section VI.  After all transmissions have been made, the inner decoder creates an estimate of the current outer code codeword by processing $\by$, $N$ entries at a time.  This process produces the following total codeword estimate
\begin{equation}
\hat{\bc} = \left[\hat{c}_{1},\hat{c}_{2},\ldots,\hat{c}_{K}\right],
\end{equation}
which will be passed on to the outer decoder for final decoding.

This setup now allows for a very convenient simplification of the concatenated coding scheme.  After processing at the inner code decoder (which will be described in Section IV), the channel given in (\ref{channel}) can be seen alternatively as
\begin{equation}
\tilde{y}[i] = c_{i} + \tilde{z}[i],
\end{equation}
\noindent as seen in Fig. \ref{concode2} where the time index $i$ is related to the original channel use index as $i = kN$.  The modified noise component, $\tilde{z}[i]$, has a new variance dependent on the properties of the inner code.  In fact, the whole effect of the inner code is encapsulated in the modified noise, $\tilde{z}[i]$.  Due to the inner code being undefined at this point, we cannot go into more depth.  However, a per-component signal-to-ratio can be calculated as
\begin{equation}
SNR_{c_{i}} = \frac{E[c_{i}^2]}{E[(\tilde{z}[i])^2]} = \frac{\rho}{E[(\tilde{z}[i])^2]}.
\end{equation}
Since the noise is i.i.d., this is the same for all $c_{i}$.  This implies that $Var(\tilde{z}[i]) = \frac{\rho}{SNR_{c_{i}}}$.

This technique has converted the closed-loop problem now into an open-loop problem.  Note that this simplification is the exact same process as defining the inner code and channel together as a \emph{superchannel} as discussed in \cite{Forney}. Since the outer code is a general forward error-correction code, this equivalent mapping has greatly reduced the problem to now finding the length $N$ inner code that maximizes $SNR_{c_{i}}$ and, thus, minimize the modified noise on the channel.  In the next few sections, this will be our exact focus.

\begin{figure}
	\centering
		\includegraphics[scale = 0.7]{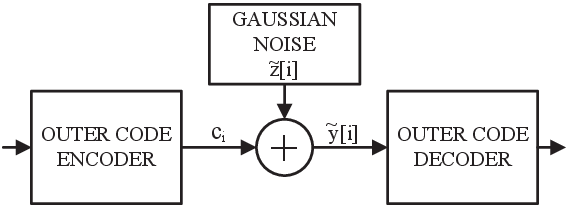}
	\caption{Simplified concatenated coding scheme ($\tilde{z}[i] \sim \cN(0,\frac{\rho}{SNR_{c_{i}}})$).}\label{concode2}
\end{figure}

\section{The Inner Code: Linear Feedback Coding}\label{sec:sys setup}
As stated above, the focus of this section is to design a length $N$ inner code that maximizes received SNR.  Since we have the availability of feedback for the inner code (see Fig. \ref{concode}) and it is a main focus of the paper, we will utilize feedback side-information at the inner code encoder.  With this setup, it is possible to employ linear feedback encoding - the advantages of which were described in the Section I.  To begin, the focus is now narrowed to only the inner code encoder/decoder pair; hence, we will only be concerned with sending and receiving one codeword of the inner code (i.e., $[s_{1}(c_{i}),s_{2}(c_{i}),\ldots,s_{N}(c_{i})]$).  This corresponds to looking at channel uses $k = ((i-1)N + 1),\ldots,iN$.  For simplicity, we study the case when $i = 1$.  To begin, the notion of general linear feedback coding is introduced.  Because of the focus of this section on the inner code and to ease with reference, we to refer to the inner code encoder as the transmitter and the inner code decoder as the receiver.

\subsection{General Linear Feedback Encoding}
\begin{figure*}
	\centering
		\includegraphics[width = 15cm, height = 5cm]{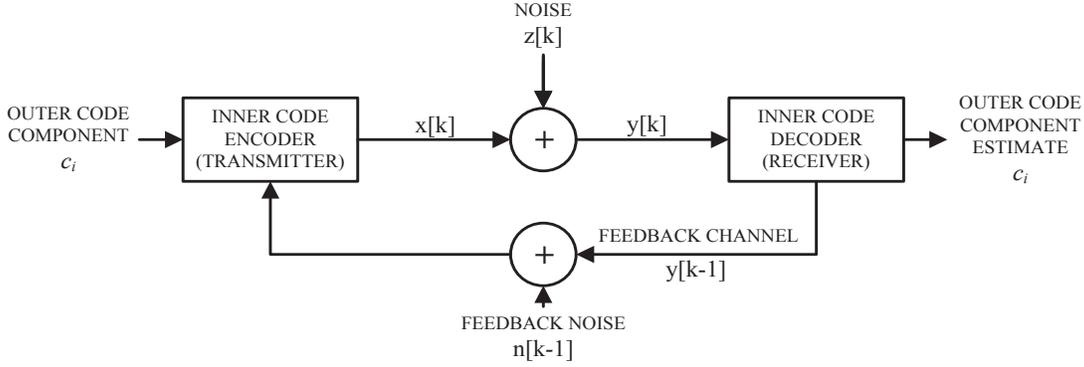}
	\caption{A transmitter/receiver pair over an AWGN channel with noisy feedback.}\label{feedback4}
\end{figure*}
In this section, we introduce the general framework of a linear feedback coding scheme in a linear algebraic formulation (similar to \cite{Butman}).  A feedback channel allows the transmission of data from the receiver back to the transmitter.  Considering the system in Fig. \ref{feedback4}, we see that such a link is available with unit delay and additive noise.  As in Section II, at channel use $k = 1,2,\ldots,N$, $x[k]$ is sent from the transmitter across an AWGN channel and the receiver receives
\begin{equation}
y[k] = x[k] + z[k],
\label{system}
\end{equation}
\noindent where $\left\{z[k]\right\}$ are i.i.d. such that each $z[k] \sim \cN(0,1)$.  Because of the feedback channel, the transmitter also has access to side-information.  In this case, we assume the side-information to be the past values of $y[k]$ corrupted by additive noise, $n[k]$.  We assume that $\left\{n[k]\right\}$ are i.i.d. such that $n[k] \sim \cN(0,\sigma^2)$ and $\left\{n[k]\right\}$ are independent of $\left\{z[k]\right\}$ .  Since we are designing an encoding scheme that will utilize feedback, $x[k]$ is encoded at the transmitter using the noisy side information $\left\{y[1]+n[1],y[2]+n[2],\ldots,y[k-1]+n[k-1]\right\}$.  By removing the known transmitted signal contribution, this is equivalent to encoding with side information $\left\{z[1] + n[1], z[2] + n[2],\ldots,z[k-1]+n[k-1]\right\}$.

We now describe a general coding scheme that utilizes this channel and feedback configuration.  The goal of the coding scheme is to reliably send a component of the outer code codeword, $c_{i}$, from transmitter to receiver across an additive noise channel using $N$ channel uses.  However, to broaden the applications of the developed scheme, we look sending a general message, $\theta \in \mathbb{R}$, instead of specifically $c_{i}$.  This is possible due to the independent operation of the inner code from the outer code.  We assume the message symbol $\theta$ is chosen from the set $\Theta = \left\{\theta_{1},\theta_{2},\ldots,\theta_{M}\right\}$ where $M$ is the number of symbols and each symbol is equally-likely.  Furthermore, we assume that $\theta$ is zero mean and that the second moment of $\theta$, $E[\theta^2]$, is known.  Due to the fact that only received SNR and rate of transmission calculations will be performed, the above description of the source alphabet proves sufficient.

With this set-up, the input to the receiver can be written as
\begin{equation}
	\by = \bx + \bz,
	\label{yeq}
\end{equation}
\noindent where, as above, the notation $\bx$ refers to $\bx = \left[x[1], x[2], \ldots, x[N]\right]^{T}$.
Because of the total average transmit power constraint (\ref{allpower}), the transmitted power of the signal $\bx$ (for $N$ transmissions) is bounded by
\begin{equation}
E[\bx^T\bx] \leq N\rho.
\label{power}
\end{equation}
\noindent The output of the transmitter $\bx$ is given as
\begin{equation}
	\bx = \bF(\bz + \bn) + \bg\theta,
	\label{xeq}
\end{equation}
\noindent where $\bg \in \mathbb{R}^{N}$ is a unit vector and $\bF \in \mathbb{R}^{N \times N}$ is a matrix called the \textit{encoding matrix}.  $\bF$ is of the form
\[\bF = \left[
\begin{array}{cccc}
0 &\cdots&\cdots&0\\
f_{2,1}&\ddots&&\vdots\\
\vdots&\ddots&\ddots&\vdots\\
f_{N,1}&\cdots&f_{N,N-1}&0
\end{array}
\right]
\]
\noindent which is referred to as \textit{strictly} lower-triangular to enforce causality.  Taking a closer look at (\ref{xeq}), we see that this is exactly the linear processing model - each $x[k]$ is a linear function of past values of $\left\{z[k] + n[k]\right\}$ and the message, $\theta$.


Now, consider the processing at the receiver's end.  The input to the receiver $\by$ is given by (\ref{yeq}).  Using (\ref{xeq}), (\ref{yeq}) becomes
\begin{equation}
\by = \bF(\bz + \bn) + \bg\theta + \bz = (\bI + \bF)\bz + \bF\bn + \bg\theta.
\label{yeq2}
\end{equation}
\noindent After all $N$ transmissions have been made, the receiver combines all received values as a linear combination and forms an estimate of the original message, $\hat{\theta}$.  This operation is written as
\[
\hat{\theta} = \bq^T\by,
\]
\noindent where $\bq \in \mathbb{R}^{N}$ is a vector called the \textit{combining vector}.  It is now evident that a general linear feedback scheme can be completely described in terms of $\bF, \bg,$ and $\bq$.  In fact, the S-K scheme can be described in this way, but since its definition is not necessary, it will be pushed to Appendix A.

As an aside, it is important to note that up until this point, a specific decoding process has not been specified.  However, since we will be passing on the output of the inner decoder straight to the outer decoder, we choose only to perform only soft decoding; hence the estimate will be sent straight to the outer decoder without mapping it to an output alphabet.  Of course, minimum-distance decoding (and similar techniques) can be easily implemented for hard decoding.

Looking back at the transmitted signal, it proves helpful to study how much power is used sending the message and how much is dedicated to encoding noise for noise-cancellation at the receiver.  This can be examined by noting that the average transmitted power is
\begin{eqnarray*}
E[\bx^T\bx] & = & \mathrm{tr}(\bF E[(\bz + \bn)(\bz+\bn)^T] \bF^T) + \left\|\bg\right\|^2E[\theta^2]\\
					  & = & \underbrace{(1+\sigma^2)\left\|\bF\right\|^2_{F}}_{\textnormal{noise-cancellation power}} + \underbrace{E[\theta^2]}_{\textnormal{signal power}}\\
					  & \leq & N\rho,
\end{eqnarray*}

\noindent where $\left\|\bF\right\|_{F}^{2} = \displaystyle\sum_{i,j}f_{i,j}^{2}$.

Because the sum of the noise-cancellation power and signal power must be less than $N\rho$, we introduce a new variable that will be a measure of the amount of power used for noise-cancellation.  To accomplish this, let us introduce $\gamma \in \mathbb{R}$ such that $0 \leq \gamma \leq 1$.  Using the power allocation factor $\gamma$, let $E[\theta^2]$ be scaled such that
\begin{equation}
E[\theta^2] = (1-\gamma)N\rho,
\label{qcon}
\end{equation}
\noindent and $\bF$ be constrained such that
\begin{equation}
(1+\sigma^2)\left\|\bF\right\|_{F}^{2} \leq N\gamma\rho.
\label{Fcon}
\end{equation}
\noindent Until Section IV, it is now assumed that $\gamma$ is fixed.

\subsection{Optimization of Received SNR}

As in Section II, our goal is to create a scheme that maximizes the received signal-to-noise ratio.  Not surprisingly, we have chosen it to be our main performance metric.  It can be derived by noting the form of the receiver's estimate of the transmitted message.  The received signal after combining is
\begin{equation}
\hat{\theta} = \bq^T\by = \bq^T((\bI + \bF)\bz + \bg\theta + \bF\bn).
\end{equation}
\noindent It follows that the received SNR is
	\[SNR = \frac{E[\left|\bq^T\bg\theta\right|^2]}{E[\left|\bq^T(\bI + \bF)\bz + \bq^T\bF\bn\right|^2]},
\]
\begin{equation}
\phantom{abcdef}= \frac{E[\theta^2]\left|\bq^T\bg\right|^2}{\left\|\bq^T(\bI + \bF)\right\|^2 + \sigma^2\left\|\bq^T\bF\right\|^2}.
\label{SNR}
\end{equation}

It would be ideal to optimize the SNR expression over all $\bF, \bg,$ and $\bq$.  However, this method turns out to be quite intractable.  Instead, we focus on optimization by two conditional optimization techniques that maximize SNR either given $\bF$ or given $\bq$.  Since the derivations of these methods are not necessary for our discussions, their proofs are pushed to Appendix B.  Note that the following procedures are hardly groundbreaking, but are given as lemmas to aid in later reference.  The first lemma is introduced to design $\bF$ to maximize the received SNR for a given $\bq$.

\begin{lemma}
Given a combining vector $\bq$ and the power constraint given in (\ref{Fcon}), the $\bF$ that maximizes the received SNR can be constructed using the following procedure:

\begin{enumerate}
\item Define $\bq_{(i)} = \left[q_{i+1},q_{i+2}, \ldots, q_{N}\right]^T$ where $1 \leq i \leq N - 1,$

\item Construct the entries of $\bF$, $f_{i,j}$, as
\[
f_{i,j} = \left\{
\begin{array}{lr}
-\frac{q_{i}q_{j}}{(1+\sigma^2)\left\|\bq_{(i)}\right\|^2 + \lambda},& i>j\\
0,& i\leq j
\end{array}\right.
\]

\noindent where $\lambda \in \mathbb{R}$ is the smallest $\lambda \geq 0$ such that $\left\|\bF\right\|^2_{F}  \leq  (1+\sigma^2)^{-1}N\gamma\rho$.

\end{enumerate}

\end{lemma}

\noindent The next lemma provides the symmetrical result; it constructs a $\bq$ that maximizes the received SNR given $\bF$.

\begin{lemma}
Given an encoding matrix, $\bF$, the $\bq$ that maximizes the received SNR can be found by letting $\bq$ be the eigenvector vector of $(\bI + \bF)(\bI + \bF)^T + \sigma^2 \bF\bF^T$ that corresponds to its minimum eigenvalue.
\end{lemma}

\noindent Note that this lemma has well-known analogous results in estimation theory.  In brief, the optimal $\bq$ is created by forming the projection
\begin{equation}
\bq^T = \frac{\bg^T\bC^{-1}}{\bg^T\bC^{-1}\bg},
\label{estim}
\end{equation}
\noindent where $\bC = (\bI+\bF)(\bI+\bF)^T + \sigma^2\bF\bF^T$.  However, it can be shown that to choose $\bg$ to maximize the received SNR given $\bF$, then (\ref{estim}) implies that $\bq = \bg$ and both should be chosen as in Lemma 2.  In addition, it is possible show that given $\bF$, the definition of $\bq$ in Lemma 2 produces the minimum variance unbiased (MVU) estimator.  To illustrate, the variance of the estimator, $\hat{\theta}$ is
\[
\var(\hat{\theta}) = E[(\theta - \hat{\theta})^2] = \bq^T\left[(\bI + \bF)(\bI + \bF)^T + \sigma^2 \bF\bF^T\right]\bq.
\]
Since $\bq$ has already been chosen to minimize this quantity and it is unbiased ($E[\hat{\theta}] = \theta$), it is the MVU estimator (consequentially also the least squares estimator).

These two lemmas will now prove sufficient for developing a linear feedback scheme that maximizes the received SNR as in Section IV.

\subsection{Upper Bound on Rate and Received SNR}
Due to the development of Lemmas 1 and 2, we can now apply them to construct some interesting results on the class of linear feedback codes.  First, an upper bound on received signal-to-noise ratio is found.

The method used in Lemma 1 to maximize the received SNR focuses on minimizing the denominator of (\ref{SNR}).  It does so while also compensating for the average power constraint given in (\ref{power}).  If this constraint is relaxed to allow the denominator of the SNR to be minimized completely, it is possible to derive an upper bound on the received SNR.

\begin{lemma}
The received SNR for a linear feedback encoding scheme with feedback noise variance, $\sigma^2$, is bounded by
\begin{equation}
SNR \leq \frac{1+\sigma^2}{\sigma^2}N\rho
\label{SNRbound}
\end{equation}
\end{lemma}
\begin{proof}
Looking at the proof of Lemma 1 (in Appendix B), the goal is to maximize the received SNR by minimizing the denominator in (\ref{SNR}).  However, the average power constraint in (\ref{Fcon}) restricts the optimization problem and the solution is not optimal in a least-squares sense.  If the power constraint is removed, (\ref{bopt1}) becomes
\begin{equation}
\bb_{opt} = \underset{\bb}{\operatorname{argmin}}\left\|\bA\bb-\bq\right\|^2 + \sigma^2\left\|\bA\bb\right\|^2.
\end{equation}

This results in the solution to the least-squares problem being
\[
\bb = ((1+\sigma^2)\bA^T\bA)^{-1}\bA^T\bq.
\]
\noindent Using this $\bb$ to construct $\bF$, (\ref{min}) becomes
\begin{equation}
\left\|\bq^T(\bI + \bF)\right\|^2 = \displaystyle\sum_{i=1}^{N-1}\left(q_{i} - \frac{q_{i}}{1+\sigma^2}\right)^2 + q^{2}_{N}
\label{min2}
\end{equation}
\begin{equation}
\geq \left(\frac{\sigma^2}{1+\sigma^2}\right)^2.
\label{min3}
\end{equation}

\noindent Similarly, the other noise term is
\begin{equation}
\left\|\bq^T\bF\right\|^2 = \displaystyle\sum_{i=1}^{N-1}\left(\frac{q_{i}}{1+\sigma^2}\right)^2 + q^{2}_{N}
\label{min4}
\end{equation}
\begin{equation}
\geq \frac{1}{(1+\sigma^2)^2}.
\label{min5}
\end{equation}

Using these two results, the received SNR, using (\ref{SNR}), can be written as
\begin{eqnarray}
SNR &\leq& \frac{E\left[\theta^2\right]}{\left(\frac{\sigma^2}{1+\sigma^2}\right)^2 + \frac{\sigma^2}{(1+\sigma^2)^2}}\\
 &=& \frac{1+\sigma^2}{\sigma^2}E\left[\theta^2\right] \\
 &\leq& \frac{1+\sigma^2}{\sigma^2}N\rho \label{SNRupper}
\end{eqnarray}
\end{proof}

In \cite{Butman}, another upper bound is given for linear feedback schemes with noise on the feedback channel.  Using the notation consistent with the above formulations, the Butman bound can be given by:
\begin{eqnarray}
SNR &\leq& E\left[\bx^T\bx\right] + \frac{1}{\sigma^2}E\left[\by^T\by\right],\\
    &=& \frac{1+\sigma^2}{\sigma^2}N\rho + 2\|\bF\|_{F}^{2} + \frac{N}{\sigma^2}.
\end{eqnarray}
\noindent Since the last two terms in the inequality are strictly greater than zero, this bound is strictly greater than (\ref{SNRbound}), implying a helpful tightness in the new bound in Lemma 3.

Suppose that we now allow the size of the symbol set, $\Theta = \{\theta_{1},\theta_{2},\ldots,\theta_{M}\}$, to be a function of the blocklength (i.e., $M^{(N)}$).  The rate in bits per channel use of our linear encoding is defined as $r^{(N)} = \log_2(M^{(N)})/N$ (Note that $r$ is used instead of $R$ to emphasize that this only applies to the inner code).  A rate $r = \lim_{N\rightarrow \infty}r^{(N)}$ is said to be achievable if the probability of error goes to zero as $N\rightarrow \infty.$  Also, if the linear feedback scheme is viewed as a superchannel (as in Fig. \ref{concode2}), the received SNR for the linear feedback scheme can be seen as the received SNR for the superchannel. Thus, the capacity of the superchannel is $\frac{1}{2}\log_{2}(1+SNR)$, where $SNR$ is the received SNR for the linear feedback scheme in use.  Now using the SNR bound result, we can construct an alternate proof of Proposition 4 given in \cite{Kim2}.
\begin{lemma}
Given any linear feedback coding scheme of rate $r$ over an AWGN channel with noisy feedback, if $r$ is achievable then  $r=0.$
\end{lemma}

\begin{proof}
As above, if regarding the linear feedback scheme over the AWGN channel as a superchannel, the capacity is
\begin{equation}
C = \frac{1}{2}\log_2{(1 + SNR)},
\end{equation}
\noindent where $SNR$ is the received SNR of the linear feedback scheme.  Then, any achievable rate $r$ must satisfy
\begin{eqnarray}
r &\leq& \lim_{N \rightarrow \infty}\frac{\frac{1}{2}\log_2{(1 + \frac{1+\sigma^2}{\sigma^2}N\rho)}}{N},\\
&=& 0.
\end{eqnarray}
\end{proof}

\section{A Linear Feedback Coding Scheme}
Now, we use both methods presented in Lemmas 1 and 2 as iterative optimization tools.  Using Lemma 1, we can design $\bF$ to maximize the received SNR.  We can do the same using Lemma 2 to design $\bq$.  However, it is desirable to optimize $\bq$ and $\bF$ jointly to maximize the SNR.  Consider being given an initial combining vector, $\bq^{(0)}$.  Using Lemma 1, we can design an encoding matrix $\bF^{(0)}$ to maximize the received SNR.  Now, that $\bF^{(0)}$ has been constructed, we can use Lemma 2 to further maximize the received SNR by designing $\bq^{(1)}$.  This process can be repeated until the received SNR does not increase with an iteration (i.e., we have reached a fixed point).

After repeatedly using this algorithm for different $\bq^{(0)}$ and different values of $N$ and $\rho$, a pattern emerges.  The structures of both $\bF$ and $\bq$ are the same for every scheme that maximizes the received SNR.  Using random search techniques, we were unable to find an alternate form that produced a higher received SNR.  Thus, empirically, the problem appears convex - the same result was produced independent of the random initial vector, $\bq^{(0)}$.  In the following conjecture, we propose that these structures of $\bF$ and $\bq$ give the scheme that maximizes the received SNR.

\subsection{The Feedback Scheme}
\begin{conjecture}
Consider again the system from Fig. \ref{feedback4}. Then, given the power constraints in (\ref{qcon}) and (\ref{Fcon}), the $\bF$ and $\bq$ that maximize the received SNR are of the following forms:
\begin{itemize}
\item $\bF$ is a strictly lower diagonal matrix with all entries along the diagonals being equal (also called a Toeplitz matrix),
\item $(1+\sigma^2)\|\bF\|_F^2 = N\gamma \rho$,
\item For some $\beta \in \mathbb{R}$ such that $\beta \in (0,1)$, the form of $\bq$ is
\[\bq = \sqrt{\frac{1-\beta^2}{1-\beta^{2N}}}\left[1,\beta,\beta^2,\ldots,\beta^{N-1}\right]^T.
\]
\end{itemize}
\end{conjecture}

\noindent Note that the term multiplying the vector $\bq$ is for normalization purposes.

Assuming that this form is optimal, we can solve for the optimal $\beta$ and the entries of $\bF$.

\begin{lemma}
Given the power constraints in (\ref{qcon}) and (\ref{Fcon}), $\bF$ and $\bq$ have the following definitions given the forms in Conjecture 1:

\begin{enumerate}

\item
The optimal $\beta$, $\beta_{0}$, is the smallest positive root of
\begin{equation}
\beta^{2N}-(N+(1+\sigma^2)N\gamma\rho)\beta^2 + (N-1),\label{betaeq3}
\end{equation}

\item
\[\bq = \sqrt{\frac{1-\beta_{0}^2}{1-\beta_{0}^{2N}}}\left[1,\beta_{0},\beta_{0}^2,\ldots,\beta_{0}^{N-1}\right]^T,
\]

\item
\[\bF = \left[
\begin{array}{ccccc}
0 & & \cdots & & 0\\
-\frac{1-\beta_{0}^2}{(1+\sigma^2)\beta_{0}} & 0\\
-\frac{1-\beta_{0}^2}{1+\sigma^2} & \ddots & \ddots & & \vdots\\
\vdots & \ddots\\
-\frac{1-\beta_{0}^2}{1+\sigma^2}\beta_{0}^{N-3} & \cdots & -\frac{1-\beta_{0}^2}{1+\sigma^2} & -\frac{1-\beta_{0}^2}{(1+\sigma^2)\beta_{0}} & 0
\end{array}\right].
\]

\end{enumerate}
\end{lemma}
\noindent The proof of Lemma 5 is given in Appendix C.

Because a closed-form solution of $\beta_{0}$ is not readily available, it proves very useful to define a close approximation.  Solving for $\beta$ in (\ref{betaeq3}), we get
\begin{equation}
\beta = \sqrt{\frac{\beta^{2N} + N-1}{N + (1+\sigma^2)N\gamma\rho}}.
\label{approx1}
\end{equation}
\noindent  Since $\beta \in (0,1)$ we can assume that $\beta^{2N} << 1$ which gives us the approximation (denoted $\beta_{1}$),
\begin{equation}
\beta_{0} \approx \sqrt{\frac{N-1}{N + (1+\sigma^2)N\gamma\rho}} \approx \sqrt{\frac{1}{1 + (1+\sigma^2)\gamma\rho}} \triangleq \beta_{1}.
\label{approx2}
\end{equation}
\noindent  The approximation, $\beta_{1}$, can be derived alternatively using iterative fixed point techniques.  This method also produces a bound on the deviation from $\beta_{0}$.  However, for values of $N > 5$, this approximation becomes extremely close.

It can be shown using (\ref{SNR}), that the received SNR for this scheme (now explicitly notating that the SNR is a function of $\beta$ and $\gamma$) is
\begin{equation}
SNR(\beta,\gamma) = \frac{(1+\sigma^2)N(1-\gamma)\rho}{\sigma^2 + \beta^{2(N-1)}}.
\label{SNReq}
\end{equation}

It is important to note that using $\beta_{1}$, the scheme exceeds the power constraint in (\ref{power}) by a small amount that dies away as the blocklength gets larger.  According to our power constraints, $\left\|\bF\right\|_{F}^{2} \leq (1+\sigma^2)^{-1}N\gamma\rho$.  However, using $\beta_{1}$ to build the scheme we get
\begin{equation}
\left\|\bF\right\|_{F}^{2} = \frac{\beta^{2(N-1)}}{(1+\sigma^2)^2} + (1+\sigma^2)^{-1}N\gamma\rho.
\end{equation}
\noindent Since $\beta \in (0,1)$ and $\sigma^2 \geq 0$,
\begin{equation}
\left\|\bF\right\|_{F}^{2} \stackrel{N \rightarrow \infty}{\rightarrow} (1+\sigma^2)^{-1}N\gamma\rho.
\end{equation}
\noindent Therefore, using $\beta_{1}$ in place of $\beta_{0}$ yields very little penalty at higher blocklengths and satisfies the power constraint as $N \rightarrow \infty$.

\subsection{Optimization Over Power Constraints}

Taking another look, the linear coding scheme described in the previous section can be further optimized if now we assume that $\gamma$ is not fixed.  This will give us another degree of freedom in attempting to maximize the received SNR.  Unfortunately, as stated above, a closed form expression for $\beta_{0}$ is unavailable, so we solve for the solution for power allocation using $\beta_{1}$.
\begin{lemma}
The power allocation scheme that maximizes the received SNR, using $\beta_{1}$, can be found using the following method:
\begin{enumerate}
\item Define:
	\begin{itemize}
		\item $a = \sigma^2$,
		\item $b = \rho(1+\sigma^2)$.
	\end{itemize}
\item Let the optimal $\gamma \in [0,1]$, $\gamma_{0}$, be the smallest positive root of
\begin{equation}
a(1+b\gamma)^N - Nb(1-\gamma) + (b+1),
\label{gamma}
\end{equation}
if it exists.  If not (when (\ref{nogamma}) is true), $\gamma_{0} = 0$.
\end{enumerate}
\end{lemma}

\begin{proof}
From above, the received SNR for our scheme is of the form
\begin{equation}
SNR(\beta_{1},\gamma) = \frac{(1 + \sigma^2)E[\theta^2]}{\sigma^2 + \beta_{1}^{2(N-1)}} =\frac{(1 + \sigma^2)N(1-\gamma)\rho}{\sigma^2 + \left(\frac{1}{1 + (1+\sigma^2)\gamma\rho}\right)^{N-1}}.
\label{obj}
\end{equation}
Ignoring the constants in the numerator and using the definitions in the lemma, maximizing (\ref{obj}) over $\gamma$ is equivalent to maximizing
\begin{equation}
\frac{1-\gamma}{a + (1 + b\gamma)^{-(N-1)}}.
\end{equation}
\noindent After taking the derivative and setting to zero, we get
\begin{equation}
a(1+b\gamma)^N - Nb(1-\gamma) + (b+1) = 0.
\end{equation}
\noindent Note that is possible to get no root that lies in $[0,1]$.  This occurs when
\begin{equation}
N < 1 + \frac{1}{\rho\left(1+\sigma^2\right)}\label{nogamma}
\end{equation}
In this case, the value of $\gamma$ reflects that noise-cancellation is no longer useful, and we set $\gamma$ to zero.
\end{proof}
A graph showing the behavior of $\gamma_{0}$ versus $\rho$ can be seen in Fig. \ref{gamma2} and a plot of $\gamma_{0}$ is given in Fig. \ref{gamma3}.  Note that the label \textit{linear units} is used to emphasize that the axis is plotted on a linear scale and not in dB.  The plots show the behavior of $\gamma_{0}$ with varying levels of feedback noise.  In both increasing either $\rho$ or $N$, it can be seen that $\gamma_{0}$ decays to zero eventually.  As $N$ increases, the additional use of feedback introduces more noise into the system, so at higher feedback levels the $\gamma_{0}$ will not peak as high and decay more quickly.  As $\rho$ increases, the numerator of the received SNR begins to predominate the maximization and $\gamma$ decreases to maximize $(1-\gamma)$ accordingly.

\begin{figure}
	\centering
		\includegraphics[height = 8cm,width = 9cm]{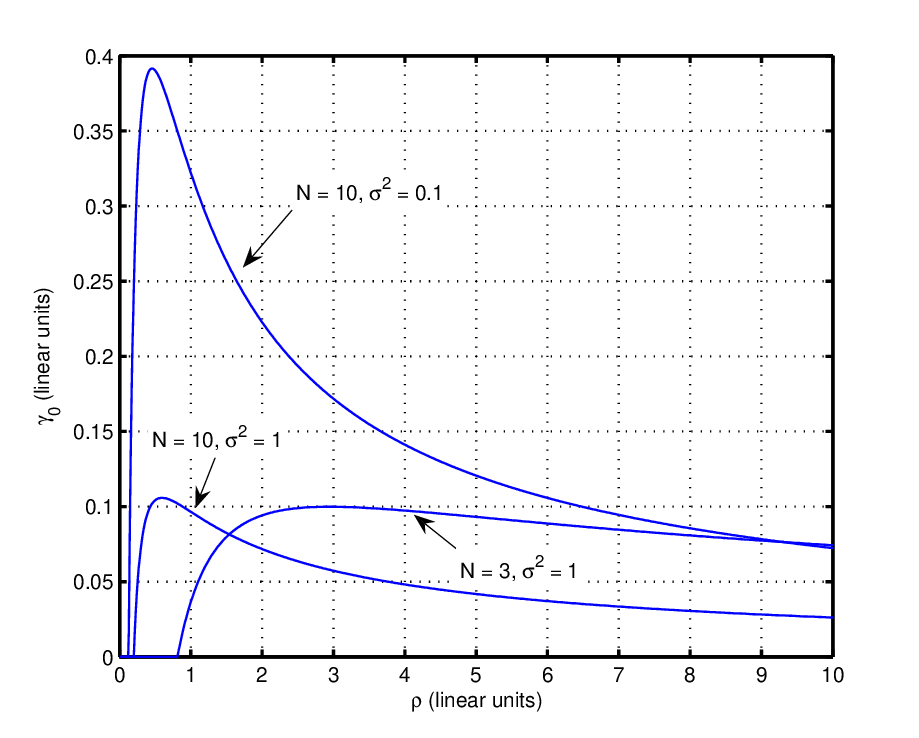}
	\caption{The behavior of $\gamma_{0}$ versus power constraint $\rho$.}
	\label{gamma2}
\end{figure}

\begin{figure}
	\centering
		\includegraphics[height = 8cm,width = 9cm]{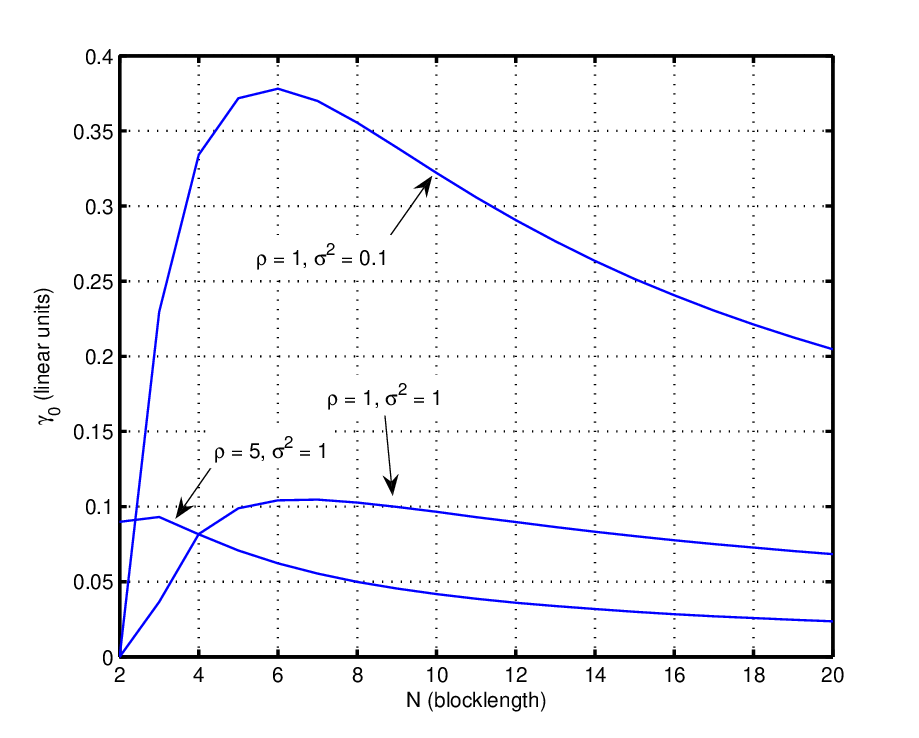}
	\caption{The behavior of $\gamma_{0}$ versus power constraint $N$.}
	\label{gamma3}
\end{figure}

An important sidenote is the behavior of this scheme (with optimal $\gamma$ and $\beta$) in the absence of feedback noise.  It turns out that as $\sigma^2 \rightarrow 0$, the method above produces the form of the solution derived in \cite{Butman} as the optimal linear feedback scheme for the AWGN channel with noiseless feedback.  However, in the presence of feedback noise, this solution noticeably differs.

\subsection{Further Analyses of the Linear Feedback Scheme}

In this section, we examine our scheme under different circumstances to derive results in related papers.

\subsubsection{Asymptotic Performance}
Using $\beta_{1}$, we can examine the asymptotic behavior of our scheme as $N \rightarrow \infty$.  If we let $\gamma = \frac{1}{\sqrt{N}}$, then the received SNR can be written as
\begin{eqnarray}
SNR\left(\beta_{1},\frac{1}{\sqrt{N}}\right) &=& \hspace{-5mm} \textstyle \frac{(1 + \sigma^2)N\left(1-\frac{1}{\sqrt{N}}\right)\rho}{\sigma^2 + \left(\frac{1}{1 + (1+\sigma^2)\frac{1}{\sqrt{N}}\rho}\right)^{N-1}},\label{first}\\
& = & \hspace{-5mm} \textstyle \frac{(1 + \sigma^2)N\left(1-\frac{1}{\sqrt{N}}\right)\rho}{\sigma^2 + \left( 1 + \frac{1}{\sqrt{N}}(1+\sigma^2)\rho\right)^{-(N-1)}},\\
& \stackrel{N \rightarrow \infty}{\rightarrow} & \frac{1+\sigma^2}{\sigma^2}N\rho. \label{asympN}
\label{asymp}
\end{eqnarray}

\noindent The received SNR of our scheme meets the upper bound in (\ref{SNRupper}) as $N \rightarrow \infty$; therefore, our scheme is asymptotically optimal.  It is worthwhile to note the choice of $\gamma$.  For this bound to appear asymptotically, $\gamma$ needs to be chosen as a function of $N$ such that $N\gamma \rightarrow \infty$ and $\gamma \rightarrow 0$ as $N \rightarrow \infty$.  Note that these constraints were motivated empirically by the behavior of $\gamma_{0}$ which is found numerically.  If $\gamma$ is not chosen within these constraints, the result (\ref{asympN}) does not apply.

\subsubsection{Binary Communications}
Now consider using our scheme to transmit a binary $(M = 2)$ symbol, $\theta$.  The probability of error, using antipodal signaling and the noise normalization currently used, can be shown to be
\begin{equation}
P_{e} = Q\left(\sqrt{SNR}\right),
\end{equation}
\noindent which as $N \rightarrow \infty$ is
\begin{equation}
P_{e} \rightarrow Q\left(\sqrt{\frac{1+\sigma^2}{\sigma^2}N\rho}\right).
\end{equation}

\noindent This expression can be bounded above by
\begin{equation}
Q\left(\sqrt{\frac{1+\sigma^2}{\sigma^2}N\rho}\right) \leq \frac{1}{2}\exp{\left[-\frac{1+\sigma^2}{2\sigma^2}N\rho\right]}.
\end{equation}
By definition, the error exponent for a given $P_{e}$ is
\begin{equation}
E(\text{binary},\rho,\sigma^2) = \lim_{N \rightarrow \infty}-\frac{1}{N} \ln\left(P_{e}\right),
\end{equation}
\noindent which in our case is
\begin{equation}
\hspace{-2mm} E(\text{binary},\rho,\sigma^2) = \hspace{-1mm}\lim_{N \rightarrow \infty}-\frac{1}{N}\ln{\left(\frac{1}{2}\exp{\left[-\frac{1+\sigma^2}{2\sigma^2}N\rho\right]}\right)}.
\end{equation}
This exponent simplifies to
\begin{equation}
E(\text{binary},\rho,\sigma^2) = \frac{(1+\sigma^2)\rho}{2\sigma^2}.
\end{equation}
\noindent This result meets the upper bound of the error exponent found in \cite{Kim2} and therefore shows that our scheme asymptotically achieves the highest rate of decay of probability of error as a function of $N$.  An illustration of this can be seen in Fig. \ref{eexp}.  This simulation was run with with exact values of $\beta_{0}$ and $\gamma_{0}$ which were found numerically.

\begin{figure}
	\centering
		\includegraphics[height = 8cm,width = 9cm]{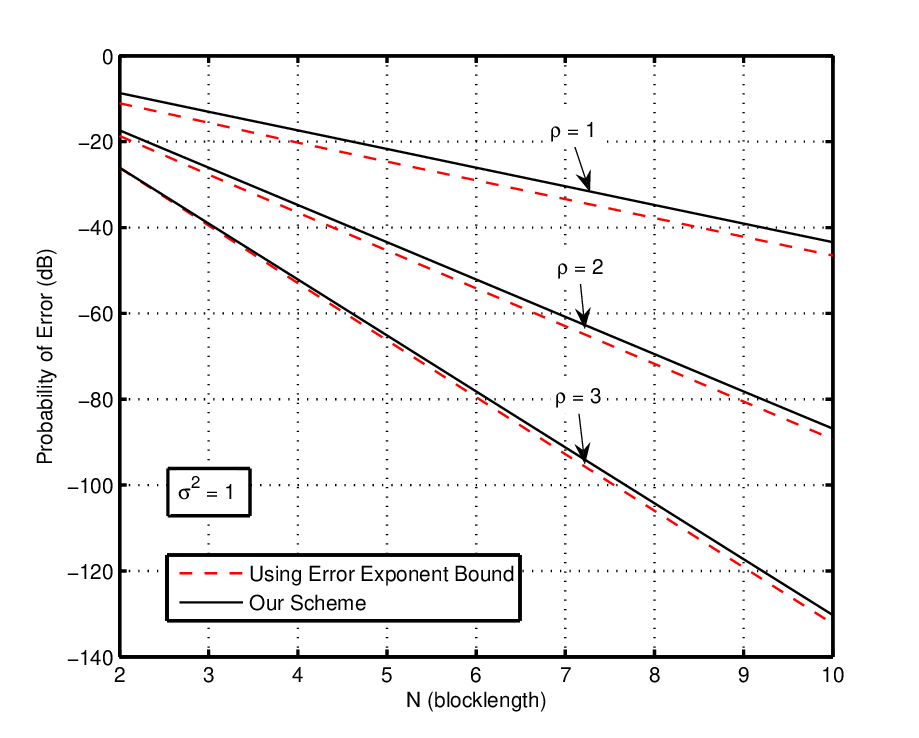}
	\caption{Comparison of the probability of error for binary transmission of the new scheme and the error exponent upper bound given in \cite{Kim2}.}
	\label{eexp}
\end{figure}
In \cite{Kim2}, a three-phase scheme is proposed that achieves this error exponent.  In brief, the message is transmitted in the first phase and, using feedback, the transmitter decides whether the receiver made the right decision.  The transmitter will then send one bit to the receiver stating whether the first transmission was a \textit{success} or a \textit{failure}.  If the transmitter decides the receiver made a wrong decision, it declares a \textit{failure} and retransmits a high-power version of the original message; otherwise, it declares a \textit{success} and does nothing.  It is important to note that this one-bit retransmission scheme was proposed in a general setting and was not constricted to binary transmissions.

\section{Simulations for the Inner Code}
We now present simulations to demonstrate the performance gains from our scheme and also the effects of feedback noise.
\subsection{Linear Feedback Comparisons}
In this section, the performance of the proposed linear scheme is compared with the Schalkwijk-Kailath (S-K) scheme (as discussed in Section I) under different circumstances.  The first simulation (Fig. \ref{skour}) plots the received SNRs for both our scheme and the S-K scheme versus the transmit SNR, $\rho$ without optimized power allocation.  The value of the optimal $\beta$, $\beta_{0}$, was found numerically and used to construct our scheme.  The feedback channel noise has variance $\sigma^2 = 0.01$.  Since the power allocation was not optimized, both schemes are using $\gamma = \frac{N-1}{N}$ (the value as given in the S-K scheme).  As can be seen, with these assumptions, our scheme shows an approximately 2 dB gain over the S-K scheme in the low $\rho$ regions $(\rho \approx 1)$.  Note that the $\rho$ axis is not in dB but a linear scale to help show the difference in performance.

\begin{figure}
	\centering
		\includegraphics[height = 8cm,width = 9cm]{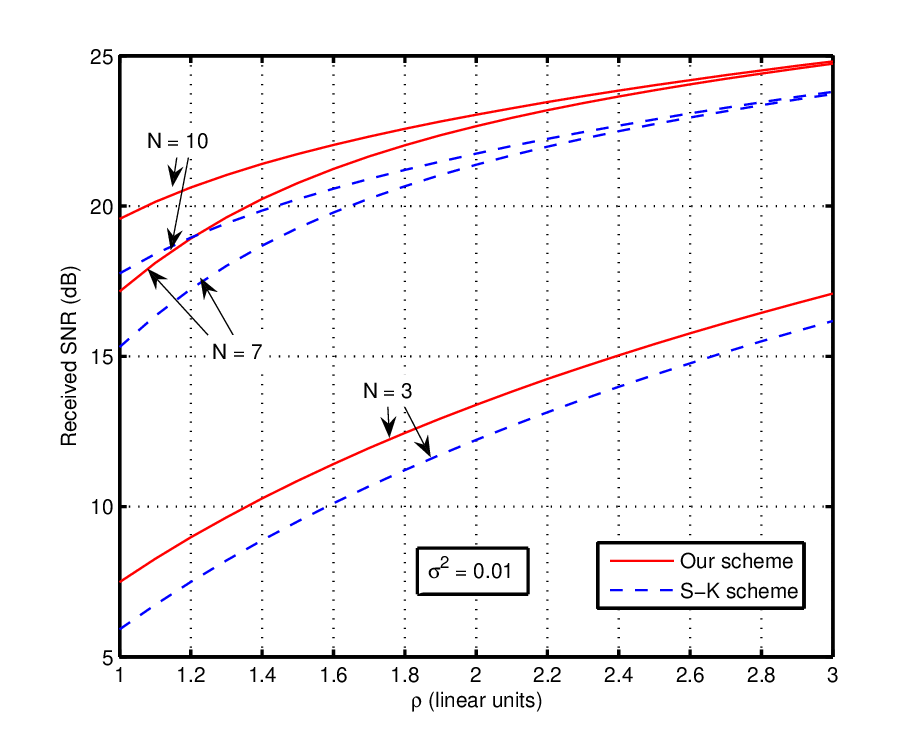}
	\caption{Comparison of the new scheme and S-K scheme with low feedback noise (without power optimization).}
	\label{skour}
\end{figure}

The next simulation (Fig. \ref{fig:sigma}) compares again the received SNR of the two schemes but for higher feedback noise $(\sigma^2 = 3)$ without power optimization ($\gamma = \frac{N-1}{N}$).  This shows quite a difference from the low feedback noise case.  Both schemes suffer a drop in performance, yet the separation between the two schemes is larger.  Another difference worth noting is the saturation of both schemes based on blocklength.  At higher feedback noise levels, blocklength does not greatly affect the performance as can be seen by the grouping of both sets of curves.  In fact, this phenomenon is due to the fact that we are using $\gamma = \frac{N-1}{N}$.  If we look at the received $SNR(\beta_{1},\frac{N-1}{N})$ for our scheme as $N \rightarrow \infty$, we can see that
\begin{eqnarray}
SNR\left(\beta_{1},\frac{N-1}{N}\right) \hspace{-5mm} &=& \hspace{-3mm}\textstyle\frac{(1 + \sigma^2)N\left(1-\frac{N-1}{N}\right)\rho}{\sigma^2 + \left(\frac{N-1}{N + (1+\sigma^2)\frac{N-1}{N}\rho}\right)^{N-1}},\\
\hspace{-5mm}& = & \hspace{-4mm}\textstyle\frac{(1 + \sigma^2)\rho}{\sigma^2 + \left( \frac{N}{N-1} + \frac{(1+\sigma^2)\rho}{N}\right)^{-(N-1)}},\\
\hspace{-5mm}& \stackrel{N \rightarrow \infty}{\rightarrow} & \frac{1+\sigma^2}{\sigma^2}\rho. \label{sksnrbound}
\end{eqnarray}
\noindent  This is a tight bound for the received SNR when using the S-K power allocation with our scheme.

\begin{figure}
	\centering
		\includegraphics[height = 8cm,width = 9cm]{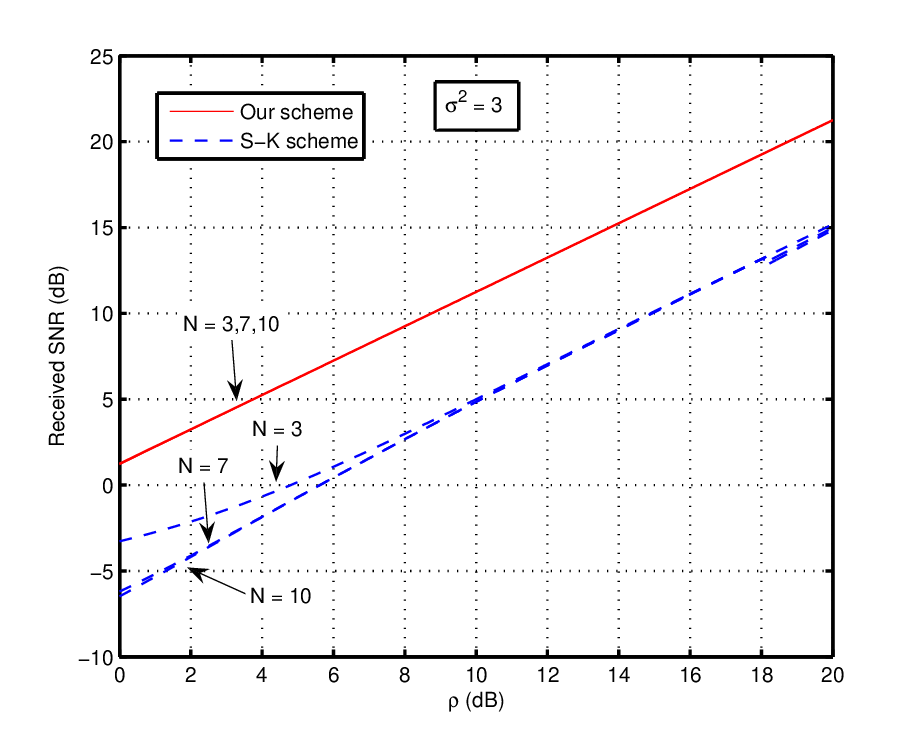}
	\caption{Comparison of the new scheme and S-K scheme with high feedback noise (without power optimization).}
\label{fig:sigma}
\end{figure}

The next figure displays the effects of optimization of power allocation.  We see from Fig. \ref{fig:power} that power allocation has greatly increased the performance of our scheme compared to the S-K scheme (still fixed at $\gamma = \frac{N-1}{N}$).  This performance increase also appears to depend on blocklength.  At $N = 3$, our scheme shows improvements in the range of 2-4 dB, but when $N = 10$, we see improvements in the range of 10 dB.  This is because it is no longer constrained by (\ref{sksnrbound}).  Because of the new choice of $\gamma$, it can now reach the $\left(\frac{1+\sigma^2}{\sigma^2}\right)N\rho$ bound.

\begin{figure}
	\centering
		\includegraphics[height = 8cm,width = 9cm]{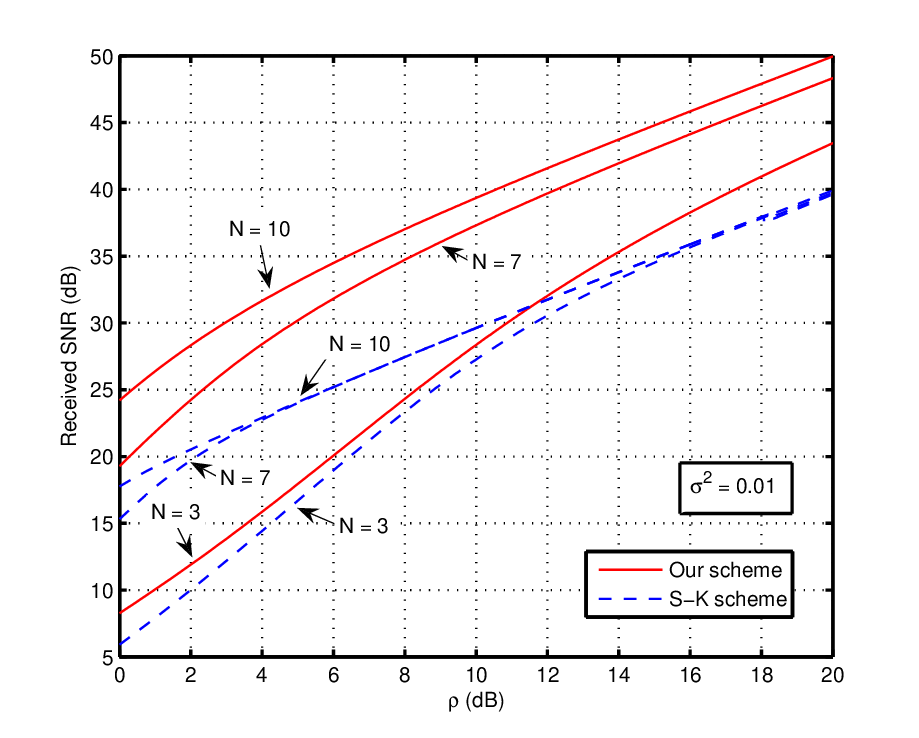}
	\caption{Optimization of power constraints provides a large improvement over the S-K scheme at low feedback noise.}
\label{fig:power}
\end{figure}

The last figure, Fig. \ref{sksigma}, shows how the received SNR of both schemes behaves with increasing feedback noise.  As is evident in the figure, the proposed linear feedback scheme is much more resilient to the effect of growing feedback noise.  Power allocation was optimized in this simulation and the average transmit power is $\rho = 1$.

\begin{figure}
	\centering
		\includegraphics[height = 8cm,width = 9cm]{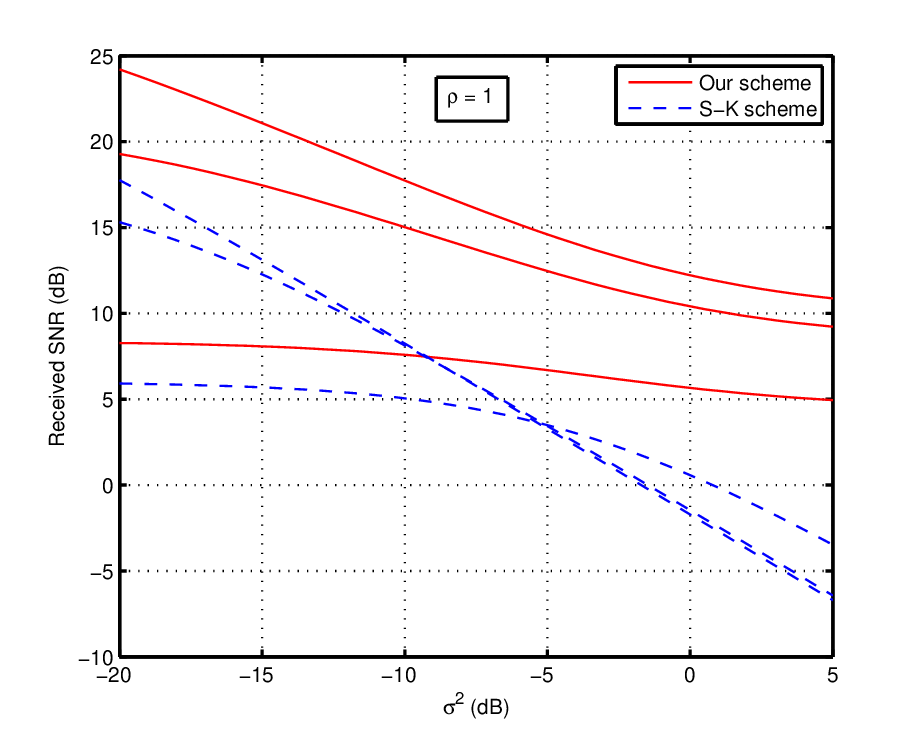}
	\caption{The proposed linear feedback scheme displays resilience in the presence of increasing feedback noise.}
\label{sksigma}
\end{figure}

\section{The Concatenated Coding Scheme}
Now that an appropriate inner code has been designed, it is possible to evaluate the performance of the total concatenated code.  For the following derivations, it is still assumed that the outer code is a general error-correction code.  In the next two sections, the error exponent for the concatenated code scheme is studied as the error exponent is an important measure of performance.  Upper and lower bounds for the error exponent are derived to illustrate the advantages of implementing feedback.

\subsection{Feedback Error Exponent Lower Bound}

The goal of this section is to find a lower bound on the reliability function for the closed-loop concatenated scheme.  To do this, we consider the best possible use of the proposed linear feedback scheme as an inner code.  To begin, let our choices for $\beta$ and $\gamma$ both be optimal such that $\beta = \beta_{0}$ from Lemma 5 and $\gamma = \gamma_{0}$ from Lemma 6 (i.e., $E[\bth^T\bth] = KN(1-\gamma_{0})\rho$).  As discussed in Section II, the problem can now be transformed into designing a $K$ channel use code for a non-feedback AWGN channel with received SNR

\begin{equation}
SNR(N,\sigma^2,\rho) = \frac{(1+\sigma^2)N(1-\gamma_{0})\rho}{\sigma^2 + \beta_{0}^{2(N-1)}},
\end{equation}
\noindent where $SNR$ is now only a function of $N$, $\sigma^2$, and $\rho$ (implicitly both $\gamma_{0}$ and $\beta_{0}$ are also functions of $N$, $\sigma^2$, and $\rho$).

Utilizing this non-feedback channel, we will now derive the error exponent expression using the open-loop reliability function.  The open-loop reliability function is defined as the rate of decay of probability of error for the best possible length $K$ coding sequence across a non-feedback channel or
\begin{equation}
E_{NoFB}(R;P) = \mathop{\mathrm{lim\hspace{1mm}sup}}_{K\rightarrow \infty}-\frac{1}{K}\ln P_{e}(R;P),
\end{equation}
\noindent coding at a rate of $R$ (bits/channel use) with a received signal-to-noise ratio $P$ and achieving a probability of error of $P_{e}(R;P)$.  Now, implementing the optimal open-loop code as the outer code over the new non-feedback channel, we achieve an open-loop error exponent of
\begin{equation}
\frac{1}{N}E_{NoFB}\left(NR;SNR(N,\sigma^2,\rho)\right).
\label{ee}
\end{equation}
The rate scaling by $N$ is due to the fact that our total blocklength has increased by a factor of $N$, but at the same time, we can only send a new symbol every $N$ channel uses.  Also, because of this structure, a trade-off in error exponent performance arises as the value of $N$ varies.  $SNR(N,\sigma^2,\rho)$ grows with increasing $N$ which is favorable, but, simultaneously, the rate increases and the factor of $\frac{1}{N}$ decreases with increasing $N$ - both adversely affecting the error exponent.  Because of this trade-off we will now define the optimal $N$, $N^*$, that achieves the highest value of the error exponent,
\begin{equation}
N^* = \argsup_{N = 1,2,\ldots} \frac{1}{N}E_{NoFB}\left(NR;SNR(N,\sigma^2,\rho)\right).
\end{equation}
Using (\ref{ee}), we can now define the closed-loop concatenated code error exponent, $E_{FB}$, by
\begin{equation}
E_{FB}\left(R;P,\sigma^2\right) = \frac{1}{N^{*}}E_{NoFB}\left(N^{*}R;SNR(N^*,\sigma^2,\rho)\right).
\label{errexp}
\end{equation}
\noindent With this result, we can now examine the concrete bounds on the error exponent for the concatenated scheme and thus create bounds for feedback error exponent.  For all rates below capacity, we can employ the random coding lower bounds \cite{shannon2} on the error exponent.
\begin{lemma}
If we first define:
\begin{eqnarray*}
R_{1} &=& \frac{1}{2}\ln{\left(\frac{1}{2} + \frac{1}{2}\sqrt{1 + \frac{P^2}{4}} \right)},\\
R_{2} &=& \frac{1}{2}\ln{\left(\frac{1}{2} + \frac{P}{4} + \frac{1}{2}\sqrt{1 + \frac{P^2}{4}} \right)},\\
C &=& \frac{1}{2}\ln{(1+P)}.
\end{eqnarray*}
Then, we can write the concatenated code error exponent lower bound as in (\ref{Efb_big}),
\begin{figure*}
\begin{equation}
E_{FB}(R,P,\sigma^2) \geq \left\{\begin{array}{l r}
\frac{\rho}{4}\frac{(1+\sigma^2)(1-\gamma_{0})}{\sigma^2 + \beta_{0}^{2(N^*-1)}}(1 - \sqrt{1 - e^{-2N^*R}}),& 0\leq R \leq R_{1}\\
\frac{\rho}{4}\frac{(1+\sigma^2)(1-\gamma_{0})}{\sigma^2 + \beta_{0}^{2(N^*_{R_{1}}-1)}}(1 - \sqrt{1 - e^{-2N^*R_{1}}}) + \frac{R_{1}}{N^*} - R, & R_{1} < R \leq R_{2}\\
\frac{1}{N^*}E_{sp}(N^*R,\frac{(1+\sigma^2)N^*(1-\gamma)\rho}{\sigma^2 + \beta_{0}^{2(N^*-1)}}),& R_{2} < R \leq C\\
\end{array}
\right.\label{Efb_big}
\end{equation}
\end{figure*}
\noindent where $E_{sp}(R,P)$ is the sphere packing bound given in \cite{shannon2}, $N^*$ is chosen at each value of $R$ to maximize the bound,  and $N^*_{R_{1}}$ is the value of $N^*$ chosen at $R_{1}$.  The optimal inner code blocklength, $N^*$ should scale as
\begin{equation}
N^* = O\left(\frac{C}{R}\right).\label{orderbnd}
\end{equation}
\noindent In addition, $N^*$, for low rates, can be approximated by
\begin{equation}
\hspace{-3mm}N^* \approx \left\lfloor\mathrm{root}\left[2(N^{3/2}-N) = \frac{\cos(\theta_{N})(1-\cos(\theta_{N}))}{R\sin^2(\theta_{N})}\right]\right\rfloor,
\label{nstar}
\end{equation}
\noindent where $0 \leq R \leq R_{1}$ and $\theta_{N} = \arcsin{(-NR)}$.
\end{lemma}
\begin{proof}
The bounds given in the lemma are a direct application of the random coding lower bounds in \cite{shannon2}.  The approximation for $N^*$ is derived as follows.  To avoid exceeding capacity, the constraint $N^*R < C$ must be imposed.  With this constraint, we can now build an approximation by looking at the expression for low rates (i.e., $0 \leq R \leq R_{1}$).  The SNR expression imposed by the use of the linear scheme is quite difficult to maximize over $N$ due to the reliance on $\beta$ and $\gamma$; therefore, to approximate it, we can replace it with an approximation that only relies on $\gamma$ and set $\gamma = \frac{1}{\sqrt{N}}$ as in Section IV.C.  Then, we can write
\begin{equation}
N^*_{approx} = \argmax_{N}\frac{(1+\sigma^2)\rho(1-\frac{1}{\sqrt{N}})}{\sigma^2}(1-\sqrt{1-e^{-2NR}}),
\end{equation}
and find the ``optimal" $N$ (by differentiating and setting the derivative to zero) is given by the root of
\begin{equation}
2(N^{3/2} - N) = \frac{\cos(\theta_{N})(1-\cos(\theta_{N}))}{R\sin^2(\theta_{N})},
\end{equation}
where $\theta_{N} = \arcsin(-NR)$.  The floor operation, $\lfloor \cdot \rfloor$, is used to keep $N^*$ an integer and to avoid violating (\ref{orderbnd}).
\end{proof}
Lemma 7 gives explicitly the random coding lower bounds for the concatenated coding scheme.  For completeness, the sphere-packing bound \cite{shannon2} will now also be defined.  If we first let $\theta(R) = \arcsin{e^{-R}}$, then the sphere packing bound can be given concisely as
\[
E_{sp}(\theta,P) = \frac{P}{2} - \frac{\sqrt{P} g(\theta,P) \cos(\theta)}{2} - \ln(g(\theta,P)\sin(\theta)),
\]
\[
g(\theta, P) = \frac{1}{2}\left(\sqrt{P}\cos(\theta) + \sqrt{P\cos^2(\theta) + 4}\right).
\]
The error exponent lower bounds as given in Lemma 7 can be seen in Fig. \ref{fig:eexp2}.  Note that the label ``no feedback" refers to the error exponent of purely the outer code with no inner code.

\begin{figure}
	\centering
		\includegraphics[scale = 0.6]{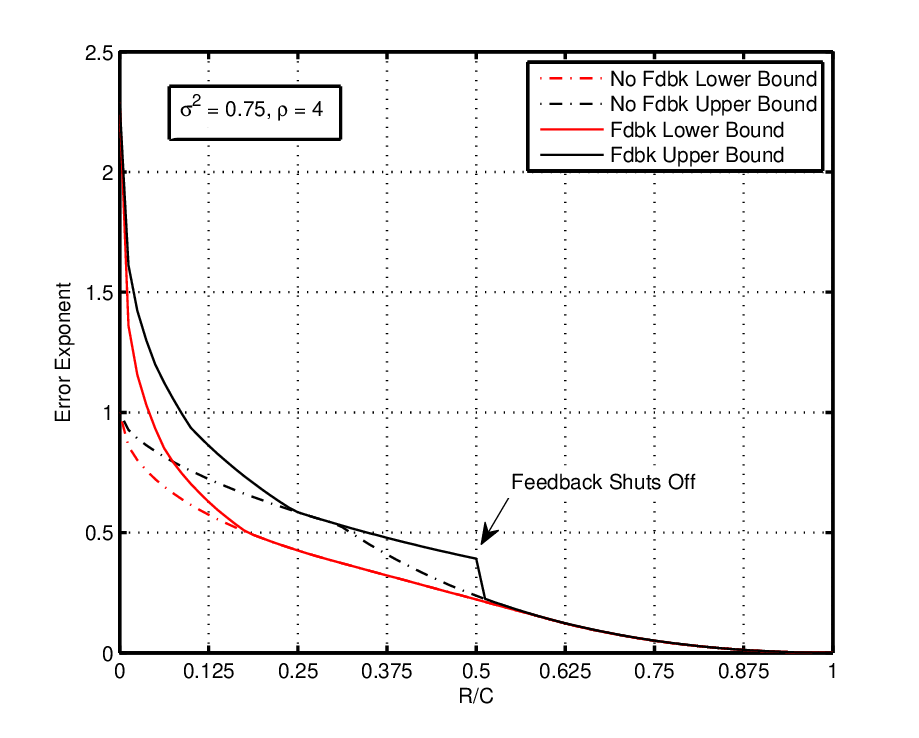}
	\caption{Error exponent bounds for non-feedback schemes and the proposed concatenated coding system.}
	\label{fig:eexp2}
\end{figure}

The $N^*$ approximation (\ref{nstar}) can be seen versus the numerically optimized $N^*$ in Fig. \ref{fig:optN}.  This gives us a rough handle on how feedback should be used (in the asymptotic sense) for the concatenated coding setup.  Namely, it should be used only at low rates but can dramatically increase the error exponent bound at these rates as seen in Fig. \ref{fig:eexp2}.

\begin{figure}
	\centering
		\includegraphics[scale = 0.62]{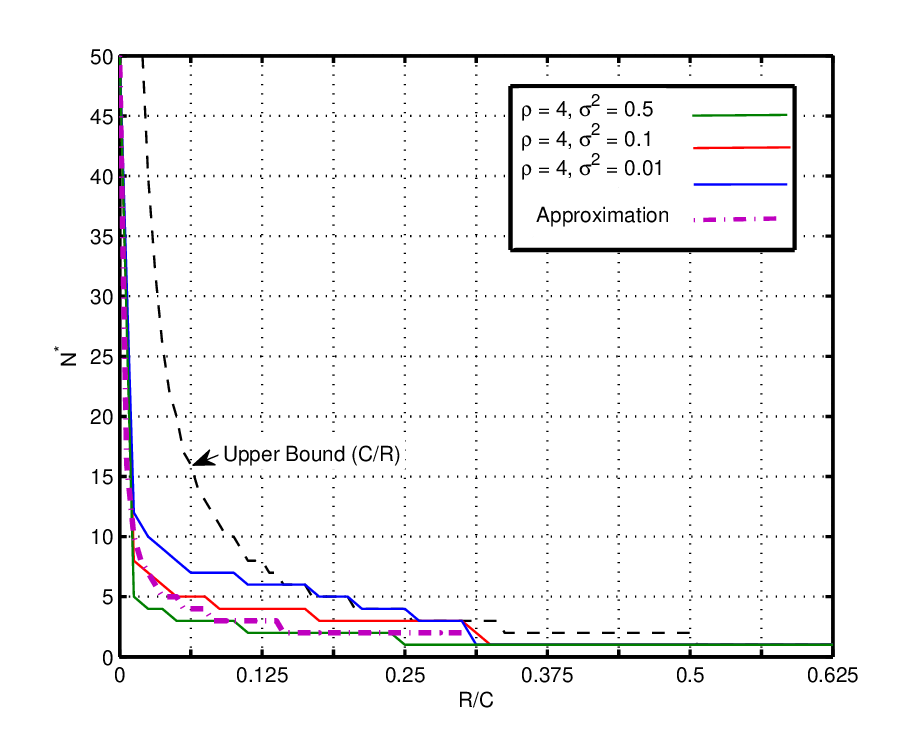}
	\caption{The $N^*$ approximation compared to the actual optimal values.}
	\label{fig:optN}
\end{figure}

\subsection{Feedback Error Exponent Upper Bound and Special Cases}
Just as important as investigating the effect of feedback on error exponent lower bounds is the effect on the upper bounds.  In this case, we employ the use of two well-known error exponent upper bounds, the minimum-distance upper bound, $E_{md}(R,P)$ \cite{Kaba,Ashik}, in conjunction with the sphere-packing bound \cite{shannon2}.  Note that the sphere-packing bound gives the exact expression for the error exponent in the $R_{2} \leq R \leq C$ region.  For reference, the minimum-distance bound can be given as
\begin{equation}
E(R,P) \leq E_{md}(R,P)  = \frac{P}{8}d^2(R),
\end{equation}
where $d^2(R)$ is the squared minimum distance of the code at rate $R$.  This can be given an upper bound as in \cite{Ashik} by first defining $\delta^*$ as the root of $R = (1+\delta)H\left(\frac{\delta}{1+\delta}\right)$ and $H(x) = -x\ln{x} - (1-x)\ln{(1-x)}$. Then,
\begin{equation}
d(R) \leq \frac{\sqrt{2}\left(\sqrt{1+\delta^*} - \sqrt{\delta^*}\right)}{\sqrt{\left(1+2\delta^*\right)}}.
\end{equation}
\noindent To ensure tightness, we take the minimum of both bounds at any given rate, $R$.  Hence, the function given in (\ref{Eupper_big}) was used to plot the upper bounds on error exponents in Fig. \ref{fig:eexp2}.
\begin{figure*}
\begin{equation}
E_{FB}(R,P,\sigma^2) \leq \min{\left(\frac{1}{N^*}E_{md}(N^*R,\frac{(1+\sigma^2)N^*(1-\gamma_{0})\rho}{\sigma^2 + \beta_{0}^{2(N^*-1)}}),\frac{1}{N^*}E_{sp}(N^*R,\frac{(1+\sigma^2)N^*(1-\gamma_{0})\rho}{\sigma^2 + \beta_{0}^{2(N^*-1)}})\right)}, 0 \leq R \leq C.\\\label{Eupper_big}
\end{equation}
\end{figure*}
\noindent Again, the value of $N^*$ is chosen to maximize the bound at any given $R$.  As in Fig. \ref{fig:eexp2}, the upper bound for feedback is higher than the upper bound in the absence of feedback.  This gap closes as feedback noise variance increases.  As noted earlier, when $N > \left\lfloor\frac{C}{R}\right\rfloor$ or when $R > \frac{C}{2}$, feedback should not be employed as the scheme exceeds the effective capacity of the superchannel - this is noted in the graph.

The main idea introduced by this section (and the previous), is that the implementation of feedback can allow for a new tradeoff - explicitly between rate and received SNR - that can exploited for further increases.  Of course, this tradeoff becomes less useful as the feedback noise increases, but it still creates a new degree of freedom.  Also, another conclusion it is possible to derive is that feedback is very beneficial at low rates.  This will be substantiated further by simulations in Section \ref{sec:sim2}.

Now, the error exponent of the concatenated coding scheme is given in the special cases of $R = 0$ and where feedback is no longer useful.
\begin{lemma}{Error Exponent Special Cases}
\begin{enumerate}
\item At $R = 0$, the error exponent for the above concatenated scheme is:
\[
E_{FB}(R=0,P)=\frac{1+\sigma^2}{4\sigma^2}\rho.
\]
\item For $R > \frac{1}{2}\log_{2}\left(1 + \frac{2(1+\sigma^2)(1-\gamma_{0})\rho}{\sigma^2}\right)$, the error exponent is
\[
E_{FB}(R,P) = E_{NoFB}(R,P),
\]
\noindent i.e., feedback is not used.
\end{enumerate}
\end{lemma}
\begin{proof}
\noindent At rates very close to zero, we can solve analytically for the error exponent of our scheme. It is a classic result that for $R = 0$, the following is true \cite{shannon2}:
\begin{equation}
E_{NoFB}\left(R=0;P\right) = \frac{P}{4}.
\end{equation}
This would imply (\ref{errexp}) can be written as
\begin{eqnarray}
E_{FB}\left(R=0;P,\sigma^2\right) &\geq& \frac{1}{4N}SNR(N,\sigma^2,\rho),\\
&=& \hspace{-3mm}\frac{1}{4N}\frac{(1+\sigma^2)N(1-\gamma_{0})\rho}{\sigma^2 + \beta_{0}^{2(N-1)}}.
\label{errexp2}
\end{eqnarray}
\noindent Note, however, that as $R \rightarrow 0$, we can let $N \rightarrow \infty$.  As stated earlier this implies $\gamma \rightarrow 0$ and $\beta^{2(N-1)} \rightarrow 0$ since $\beta \in (0,1)$.  This produces
\begin{equation}
E_{FB}\left(R=0;P,\sigma^2\right) = \frac{1+\sigma^2}{4\sigma^2}\rho > \frac{\rho}{4} = E_{NoFB}(R=0,P).
\end{equation}

For the second result, consider the specific case of $N = 2$.  Fortunately, when $N = 2$, we can solve analytically for $\beta$ using Lemma 2.  After some algebra, we find
\begin{equation}
\beta_{0} = \sqrt{\frac{(1+\sigma^2)\gamma\rho}{2}+1} - \sqrt{\frac{(1+\sigma^2)\gamma\rho}{2}}.
\end{equation}
\noindent Using this value of $\beta$, the received SNR is calculated to be
\begin{eqnarray}
SNR(\beta_{0},\gamma) &=& \frac{(1+\sigma^2)N(1-\gamma)\rho}{\sigma^2 + \beta_{0}^{2(N-1)}},\\
& = & \hspace{-2mm}\textstyle\frac{(1+\sigma^2)N(1-\gamma)\rho}{\sigma^2 + \left(\sqrt{\frac{(1+\sigma^2)\gamma\rho}{2}+1} - \sqrt{\frac{(1+\sigma^2)\gamma\rho}{2}}\right)^2},\\
&<& \frac{2(1+\sigma^2)(1-\gamma)\rho}{\sigma^2}.
\label{snrleq}
\end{eqnarray}
\noindent For a rate to be achievable, it must satisfy
\begin{equation}
NR \leq \log_{2}\left(1 + SNR(\beta_{0},\gamma_{0})\right),
\end{equation}
\noindent where $\gamma_{0}$ is the optimal $\gamma$ defined in Lemma 6.  Setting $N = 2$ and using (\ref{snrleq}), feedback should not be employed with our concatenated scheme if
\begin{equation}
R > \frac{1}{2}\log_{2}\left(1 + \frac{2(1+\sigma^2)(1-\gamma_{0})\rho}{\sigma^2}\right).
\end{equation}
\end{proof}

\section{Simulations for Concatenated Coding}\label{sec:sim2}
In this section, the performance of the concatenated coding system in Section VI is simulated for the cases where the outer code is an LDPC code (according to WiMAX standard) and a turbo code (according to UMTS standard).  Details for each code are given below.  These simulations were run using the Coded Modulation Library \cite{CML}.  To keep the number of channel uses consistent, the concatenated coding scheme has to implement $2^N$ modulation order versus the open-loop technique using BPSK.  Therefore, to use an inner code with 2 iterations of the proposed linear scheme, $2^2$ modulation order is used (i.e., QPSK).  This can be seen alternatively as splitting complex modulation transmission into two parallel real modulation transmissions.  Also, to ensure that both schemes use the same average power, the linear feedback scheme must be designed with a particular value of $\rho$.  In particular, if the open-loop technique uses $E_{s}$ energy per symbol to modulate, then the linear feedback code must be designed with $\rho = E_{s}$ when using $2^N$ modulation order.
\begin{figure}
  \centering
  \includegraphics[scale=0.55]{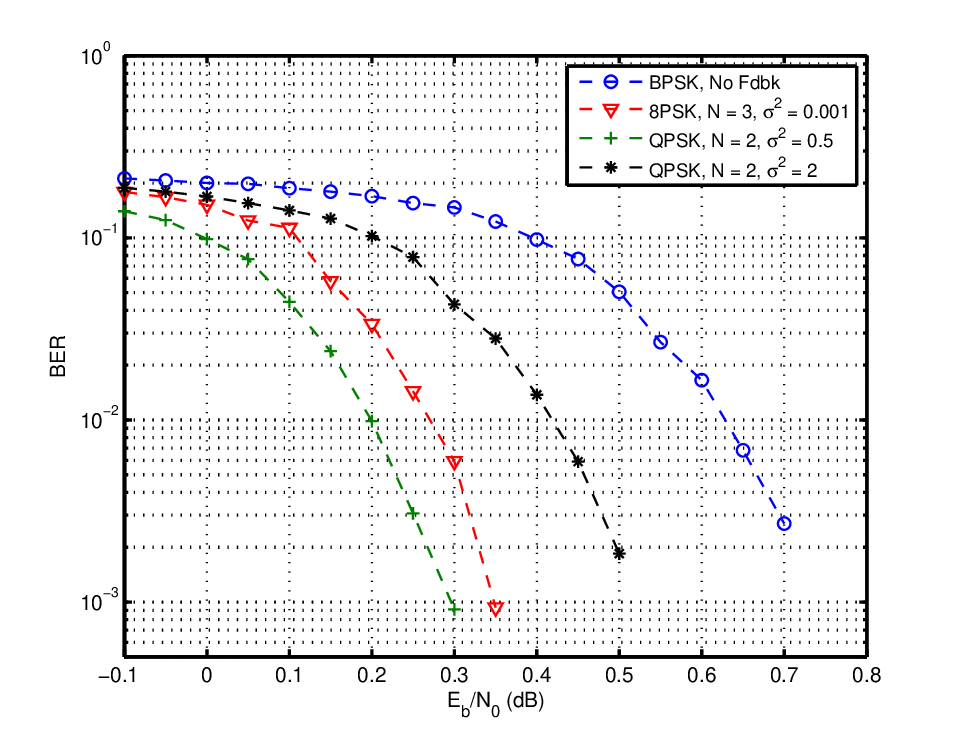}\\
  \caption{BER performance of concatenated coding scheme versus open-loop coding for UMTS turbo code ($K = 5114$ bits, Rate $= 1/3$, and 10 decoding iterations).}
\end{figure}
\begin{figure}
\centering
  \includegraphics[scale = 0.55]{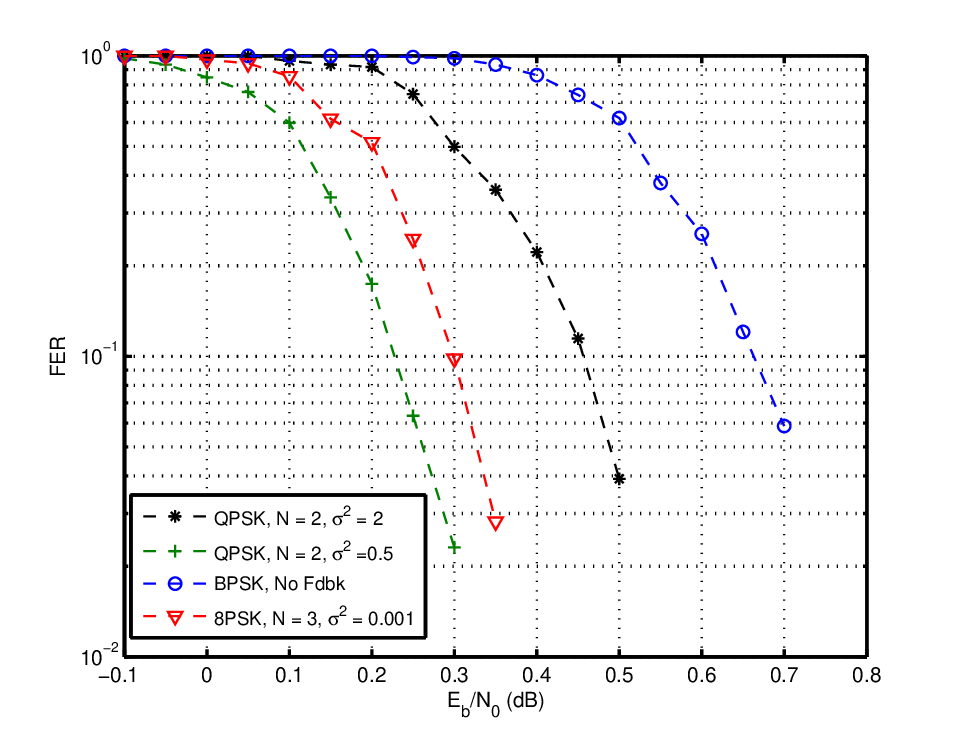}\\
  \caption{FER performance of concatenated coding scheme versus open-loop coding for UMTS turbo code ($K = 5114$ bits, Rate $= 1/3$, and 10 decoding iterations).}
\end{figure}
Fig. 13 shows that for a fixed $\frac{E{b}}{N_{0}}$, the probability of bit error is up to $0.4-0.5$ dB lower by using the concatenated coding scheme for the UMTS turbo code (Rate $1/3$, $K = 5114$ bits, 10 decoding iterations).  Note that this turbo code uses the max-log-MAP algorithm.  This gain in performance is also a function of the feedback noise variance.  As can be seen, the performance gains diminish as feedback noise increases.  The same phenomenon is apparent in Fig. 14 which is a comparison of the frame-error-rate (FER) for both techniques.  A similar improvement (up to $0.4-0.5$ dB) can be seen.

\begin{figure}
  \centering
  \includegraphics[scale=0.55]{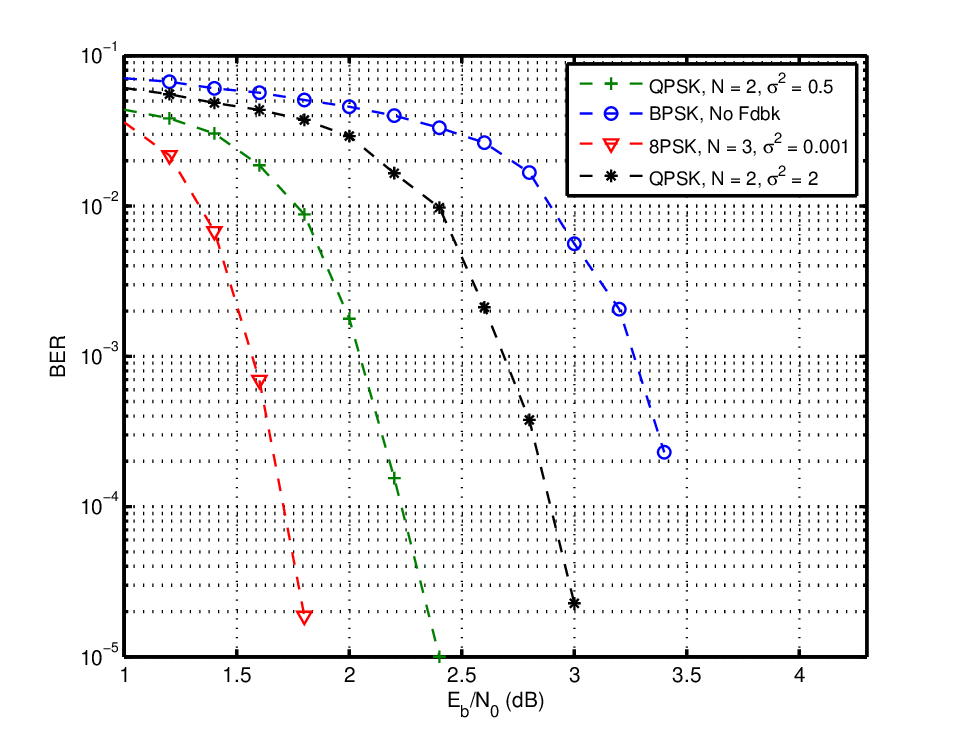}\\
  \caption{BER performance of concatenated coding scheme versus open-loop coding for WiMAX LDPC code ($K = 2304$ bits, Rate $= 5/6$, and 100 decoding iterations).}
\end{figure}
\begin{figure}
\centering
  \includegraphics[scale = 0.55]{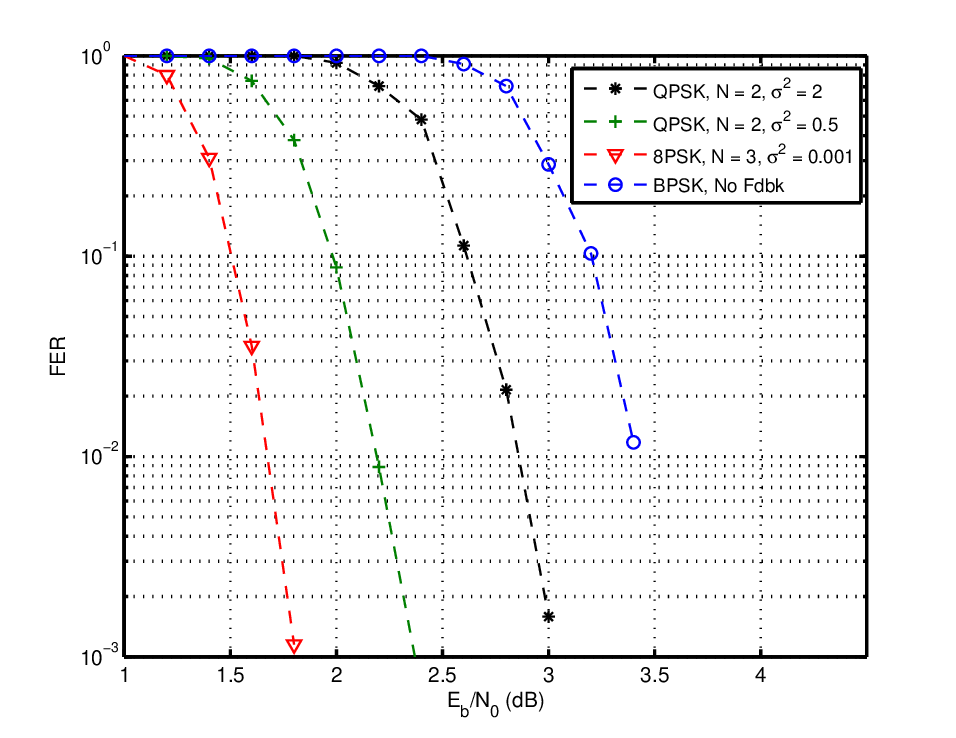}\\
  \caption{FER performance of concatenated coding scheme versus open-loop coding for WiMAX LDPC code ($K = 2304$ bits, Rate $= 5/6$, and 100 decoding iterations).}
\end{figure}

Fig. 15 display the the BER for both the open-loop and concatenated coding scheme using the LDPC code as given in the WiMAX standard (Rate $5/6$, $K = 2304$ bits, 100 decoding iterations).  Again, the concatenated code is modulated using QPSK and has an inner code of two iterations of the proposed linear feedback scheme.  The performance of the concatenated coding schemes again display lower error rates than pure open-loop techniques - displaying up to around 2 dB improvement.  The effect of feedback noise is clear as it greatly closes the gap between the two methods.  However, it is interesting to see that the $N = 3, \sigma^2 = 0.001$ performs much better when compared to the turbo code.  Fig. 16 displays the FER for both schemes which demonstrates up to around 2 dB improvement.

An interesting point introduced by extending an open-loop error-correcting code into a closed-loop concatenated code is the tradeoff between modulation order and the increase in received SNR for the channel.  If we increase the number of iterations of feedback coding, the received SNR increases.  However, simultaneously, the modulation order increases which creates a less forgiving probability of symbol error.  This tradeoff allows for a new degree of freedom in transmission schemes that can be exploited to achieve lower error rates.

\section{Conclusions}\label{sec:conc}
In this paper, we investigated a specific case of concatenated coding for the AWGN channel with noisy feedback.  The inner code was designed as a linear feedback scheme that was constructed to maximize received signal-to-noise ratio.  The performance of the linear feedback scheme was compared to another well-known feedback technique, the Schalkwijk-Kailath scheme.  The outer code was allowed to be any open-loop error correction code for ease of adaptation.  The concatenated coding scheme shows that the use of feedback can greatly increase error exponent bounds compared to pure open-loop techniques. Simulations illustrated the performance of the linear scheme and its incorporation into the concatenated coding scheme when the outer code is either a turbo code or LDPC code.

\begin{appendix}
\subsection{Schalkwijk-Kailath Coding Scheme}
The S-K scheme is a special case of the linear feedback encoding framework formulated in Section III.  When describing the S-K scheme we will ignore feedback noise $(\sigma^2 \rightarrow 0)$, since it was designed for a noiseless feedback channel.  In the S-K set-up, $\gamma = \frac{N-1}{N}$ and $\bg_{SK}$, $\bF_{SK}$, and $\bq_{SK}$ have the following definitions:

\begin{enumerate}
\item $\bg_{SK} = \left[1, 0, \ldots, 0\right]^{T}$,

\item Let $\alpha^{2} = 1 + \rho$ and $r = \sqrt{\rho}$.  Then $\bF_{SK}$ is an $N\times N$ encoding matrix given by
\[\hspace{-3mm}\bF_{SK} = \left[\begin{array}{ccccccc}
									0 & 0 & & & \cdots & & 0\\
									-r & 0\\
									\frac{-r}{\alpha} & \frac{-r^{2}}{\alpha} & 0\\
									\frac{-r}{\alpha^{2}} & \frac{-r^{2}}{\alpha^{2}} & \frac{-r^{2}}{\alpha} & 0 &&& \vdots\\
									\frac{-r}{\alpha^{3}} & \frac{-r^{2}}{\alpha^{3}} & \frac{-r^{2}}{\alpha^{2}} & \frac{-r^{2}}{\alpha} & 0\\
									\vdots & \vdots & \vdots & \vdots & \ddots & \ddots\\
									\frac{-r}{\alpha^{N-2}} & \frac{-r^{2}}{\alpha^{N-2}} & \frac{-r^{2}}{\alpha^{N-3}} && \cdots & \frac{-r^{2}}{\alpha} & 0\\
									\end{array}\right],
\]
\item
\[\bq_{SK} = \left[1,\frac{r}{\alpha^2},\frac{r}{\alpha^3},\ldots,\frac{r}{\alpha^{N}}\right]^T.
\]

\end{enumerate}
\subsection{Optimization of Received SNR}
For this optimization, let us assume that $\gamma$ is fixed and without loss of generality, $\bg$ and $\bq$ are both unit vectors.  With those assumptions, the goal at this point is to design $\bg, \bq,$ and $\bF$ to maximize (\ref{SNR}).  Looking first at the numerator, we see that we can bound $\left|\bq^T\bg\right|^2$ using the Cauchy-Schwarz inequality.  Doing this, we see that
\begin{eqnarray*}
\left|\bq^T\bg\right|^2 &\leq& \left\|\bq\right\|^2\left\|\bg\right\|^2\\
	& = & 1.
\end{eqnarray*}
\noindent This bound can be achieved by letting $\bg = \bq$.  For our purposes now, we will always assume that $\bg = \bq$, $\bF$ is restricted as in (\ref{Fcon}), and $E[\theta^2] = N(1-\gamma)\rho$.  With these conditions, the received SNR were are trying to optimize simplifies to
\begin{equation}
SNR = \frac{N(1-\gamma)\rho}{\left\|\bq^T(\bI + \bF)\right\|^2 + \sigma^2\left\|\bq^T\bF\right\|^2}.
\label{SNR2}
\end{equation}
\noindent Note also that in the S-K case, even though $\bq_{SK}$ is not a unit vector, still $\left|\bq_{SK}^T\bg_{SK}\right|^2 = 1$.

Since the numerator is now fixed, our focus now turns towards minimizing the denominator.  However, this is more complicated.  The ideal solution would be to jointly minimize the denominator over $\bq$ and $\bF$.  Unfortunately, this does not yield any feasible path towards a solution.  Instead of attempting to jointly optimize, we now derive the two conditional optimization methods used as Lemma 1 and Lemma 2.

First, consider minimizing the denominator given a combining vector $\bq$.  Since $\bq$ is given, the goal is to design $\bF$ to maximize (\ref{SNR}); therefore we should pick $\bF$ using
\begin{equation}
\begin{array}{ccc}
\hspace{-2mm}\bF_{opt} = \hspace{-4mm} &\underset{\bF}{\operatorname{argmin}}&\left\|\bq^T(\bI + \bF)\right\|^2 + \sigma^2\left\|\bq^T\bF\right\|^2,\\
&\textnormal{subject to} & \left\|\bF\right\|_{F}^2 \leq (1+\sigma^2)^{-1}N\gamma\rho,~f_{i,j} = 0,~i \leq j.
\vspace{2mm}
\end{array}
\label{Fopt}
\end{equation}

We now have sufficient background to prove Lemma 1.
\begin{proof}{(Lemma 1)}
To begin let us define the non-zero columns of $\bF$ as $\bff_{i} = \left[ \hspace{1mm} f_{i+1,i},f_{i+2,i},\ldots,f_{N,i}\right]^T$ for $1 \leq i \leq N-1$.  Now, working through the multiplication, we can rewrite
\begin{equation}
  \left\|\bq^T(\bI + \bF)\right\|^2 = \displaystyle\sum_{i=1}^{N-1}(q_{i} + \bq_{(i)}^{T}\bff_{i})^2 + q_{N}.
  \label{firstterm}
\end{equation}

At this point, it is worthwhile to remark that minimizing this sum is equivalent to minimizing the total sum given in (\ref{Fopt}).  This is due to the fact that the subspace for the solution of $\bF$ in the first term is that same as in the second term, $\sigma^2\|\bq^T\bF\|^2$.  This can be seen as both terms can be written in the form $\|\bF\bq + \bb\|^2$ where in the first term $\bb = \bq$ and in the second $\bb = 0$.  Both solutions can be carried out the same way with the assumption that $\bb \geq 0$, which is assumed.  For lack of redundancy, only the minimization of the first term is explicitly carried out.

Looking back, to minimize (\ref{firstterm}), we need to minimize $q_{i} + \bq_{(i)}^{T}\bff_{i}$ for all $i$.  This can be accomplished by designing the $\left\{\bff_{i}\right\}$ such that
\begin{equation}
\bff_{i} = - \frac{\bq_{(i)}}{\left\|\bq_{(i)}\right\|}\alpha_{i},
\label{f}
\end{equation}
\noindent where
\begin{equation}
\displaystyle\sum_{i=1}^{N-1}\alpha_{i}^{2} \leq (1+\sigma^2)^{-1}N\gamma\rho.
\label{const}
\end{equation}
\noindent The introduction of $\left\{\alpha_{i}\right\}$ is required because of the constraint, $\left\|\bF\right\|_{F}^{2} \leq (1+\sigma^2)^{-1}N\gamma\rho$.  Substituting in for the new columns of $\bF$ produces
\begin{equation}
\left\|\bq^T(\bI + \bF)\right\|^2 = \displaystyle\sum_{i=1}^{N-1}(q_{i} - \left\|\bq_{(i)}\right\|\alpha_{i})^2 + q_{N}.
\label{min}
\end{equation}

This limits the problem of designing the matrix $\bF$ to finding the $\left\{\alpha_{i}\right\}$ that minimize (\ref{min}) and satisfy (\ref{const}) - this is a norm-constrained least squares problem.  This is more evident if we let
	\[\bA = \left[\begin{array}{cccc}
									\left\|\bq_{(1)}\right\| & 0 & \cdots & 0\\
									0 & \left\|\bq_{(2)}\right\| & \cdots & \vdots\\
									0 & 0 & \ddots & 0\\
									0 & 0 & 0 & \left\|\bq_{(N-1)}\right\|\\
									0 & 0 & \cdots & 0\\
									\end{array}\right]
\] and $\bb = \left[\alpha_{1},\alpha_{2},\ldots,\alpha_{N-1}\right]^T$.  Thus, rewriting (\ref{min}), the problem of minimizing the $\left\|\bq^T(\bI + \bF)\right\|^2$ term now becomes
\[
\begin{array}{cc}
\min & \left\|\bA\bb - \bq\right\|^2.\\
\textrm{subject to} & \left\|\bb\right\|^2 \leq (1+\sigma^2)^{-1}N\gamma\rho
\end{array}
\]

\noindent Noting that $\bq^T(\bI + \bF) = (\bq - \bA\bb)^T$ and $\bq^T\bF = (-\bA\bb)^T$, we can calculate the optimal $\bb$ using
\begin{equation}
\begin{array}{ccc}
\bb_{opt} = &\underset{\bb}{\operatorname{argmin}}&\left\|\bA\bb-\bq\right\|^2 + \sigma^2\left\|\bA\bb\right\|^2,\\
& \textnormal{subject to} & \left\|\bb\right\|^2 \leq (1+\sigma^2)^{-1}N\gamma\rho.
\end{array}
\label{bopt1}
\end{equation}
\noindent At this point, the focus is now not only on the first term but taken over the whole sum in (\ref{Fopt}) as can be seen with the introduction of the $\sigma^2$ term in (\ref{bopt1}).  To solve for the optimal $\bb$ and make sure that $\left\|\bb\right\|^2 \leq (1+\sigma^2)^{-1}N\gamma\rho$, we use Lagrange multipliers.  Forming the Lagrangian, we get
\[
\begin{array}{c}
L(\bb,\lambda) = \bq^T\bq - 2\bb^T\bA^T\bq +\bb^T\bA^T\bA\bb +\hspace{1mm}\sigma^2\bb^T\bA^T\bA\bb\\
+~\lambda(\bb^T\bb - (1+\sigma^2)^{-1}N\gamma\rho).\\
\end{array}
\]
\noindent After taking the gradient with respect to $\bb$ and setting to zero, solving for the optimal $\bb$ results in
\begin{equation}
\bb_{opt} = ((1 + \sigma^2)\bA^T\bA + \lambda\bI)^{-1}\bA^T\bq,
\label{b}
\end{equation}
\noindent where $\lambda$ is chosen such that $\bb^T\bb = (1+\sigma^2)^{-1}N\gamma\rho$.  Once $\bb$ has been calculated, $\bF$ can be constructed using (\ref{f}).
\end{proof}

To prove Lemma 2, we consider the case when $\bF$ is given and we are designing $\bq$ to maximize the received SNR.  The goal now is to find $\bq$ such that
\[
\begin{array}{ccc}
\bq_{opt} = &\underset{\bq}{\operatorname{argmin}}&\left\|\bq^T(\bI + \bF)\right\|^2 + \sigma^2\left\|\bq^T\bF\right\|^2\\
&\textnormal{subject to}& \left\|\bq\right\|^2 = 1
\end{array}
\]
\noindent This problem, however, can be solved very quickly as given in the following proof of Lemma 2.

\begin{proof}{(Lemma 2)}
Let $\delta_1, \delta_2, \ldots, \delta_{N}$ be the eigenvalues of $(\bI + \bF)(\bI + \bF)^T + \sigma^2 \bF\bF^T$ such that $\delta_1 \geq \delta_2 \geq \ldots \geq \delta_{N} \geq 0$.  Then,
\[
\begin{array}{c}
\left\|\bq^T(\bI + \bF)\right\|^2 + \sigma^2\left\|\bq^T\bF\right\|^2 =\\
\vspace{-3mm}\\
\bq^T\left[(\bI + \bF)(\bI + \bF)^T + \sigma^2 \bF\bF^T\right]\bq \geq \delta_{N}.
\end{array}
\]
This bound can be achieved by letting $\bq$ be the eigenvector of $(\bI + \bF)(\bI + \bF)^T + \sigma^2 \bF\bF^T$ corresponding to $\delta_{N}$.  This choice of $\bq$ leads to $\left\|\bq^T(\bI + \bF)\right\|^2 + \sigma^2\left\|\bq^T\bF\right\|^2 = \delta_{N}$.

\end{proof}
These two conditional solutions allow for numerical optimization as discussed in Section IV.

\subsection{Proof of Lemma 5}
We now provide the proof for the structure of our linear scheme as given in Lemma 5.
\begin{proof}
To find the entries of $\bF$, let us consider entries $f_{N-1,N-2}$ and $f_{N,N-1}$ shown below:
\[
\bF = \left[
\begin{array}{ccccc}
0&&\cdots&&0\\
 f_{2,1}\\
\vdots&\ddots&\ddots& & \vdots\\
& & f_{N-1,N-2}\\
 f_{N,1}&\cdots& f_{N,N-2} & f_{N,N-1} & 0
\end{array}
\right]
\]
From the form in Conjecture 1, we should have that
\begin{equation}
f_{N-1,N-2} = f_{N,N-1}.
\label{equal}
\end{equation}
\noindent Now we use Lemma 1 to begin finding the form of $\bF$ given the exponential form of $\bq$.  Using step 3 of Lemma 1, we  compute $\bb$ as
\begin{equation}
\bb = \left[
\begin{array}{c}
\frac{\beta^0 \left\|\bq_{(1)}\right\|}{\lambda + (1+\sigma^2)\left\|\bq_{(1)}\right\|^2}\\
\\
\frac{\beta^1 \left\|\bq_{(2)}\right\|}{\lambda + (1+\sigma^2)\left\|\bq_{(2)}\right\|^2}\\
\vdots\\
\frac{\beta^{N-2} \left\|\bq_{(N-1)}\right\|}{\lambda + (1+\sigma^2)\left\|\bq_{(N-1)}\right\|^2}\\
\end{array}
\right].
\end{equation}

\noindent Now, using the definitions of the columns from step 4 of Lemma 1, we get
\begin{eqnarray}
f_{N-1,N-2} &=& \frac{-\beta^{N-2}\beta^{N-3}}{\lambda + (1+\sigma^2)\left\|\bq_{(N-2)}\right\|^2},\\
\label{F1}
f_{N,N-1} &=& \frac{-\beta^{N-1}\beta^{N-2}}{\lambda + (1+\sigma^2)\left\|\bq_{(N-1)}\right\|^2}.
\end{eqnarray}
\noindent Then, using (\ref{equal}), we solve for $\lambda$ which produces
\begin{equation}
\lambda = \frac{(1+\sigma^2)\left(\beta^2\left\|\bq_{(N-2)}\right\|^2 - \left\|\bq_{(N-1)}\right\|^2\right)}{1-\beta^2}.
\label{lambda}
\end{equation}

Since the form of $\bq$ consists of consecutive powers of $\beta$, we can state the following:
\[
\left\|\bq_{(N-2)}\right\|^2 - \left\|\bq_{(N-1)}\right\|^2  =  \displaystyle \sum_{i = N-2}^{N-1}\beta^{2i} - \sum_{i = N-1}^{N-1}\beta^{2i},
\]
\begin{equation}
 \phantom{abcdedegf} =  \beta^{2(N-2)}.
 \label{lamsim}
\end{equation}

\noindent Using the value of $\lambda$ from (\ref{lambda}) in $\bb$ and simplifying using (\ref{lamsim}) results in the $(N-2)^{th}$ component of $\bb$ being
\[
b_{N-2} = \frac{\left\|\bq_{(N-2)}\right\|\beta^{N-3}(1-\beta^2)}{(1+\sigma^2)\beta^{2(N-2)}}.
\]

\noindent Using $b_{N-2}$ to construct $\bff_{N-2}$, we find

\begin{eqnarray}
\textstyle
\bff_{N-2} &=&
\left[\begin{array}{c}
f_{N-1,N-2}\\
f_{N,N-2}
\end{array}\right],\\
&=& b_{N-2}\left(\frac{-1}{\left\|\bq_{(N-2)}\right\|}\right)
\left[\begin{array}{c}
\beta^{N-2}\\
\beta^{N-1}\\
\end{array}
\right],\\
&=& \left[
\begin{array}{c}
-\frac{1-\beta^2}{(1+\sigma^2)\beta}\\
-\frac{1-\beta^2}{(1+\sigma^2)}
\end{array}
\right].\\
\end{eqnarray}

%

\noindent Using this pattern we find that any non-zero column of $\bF$ can be written as

\[
\bff_{i} =
\left[
\begin{array}{c}
f_{i+1,i}\\
f_{i+2,i}\\
\vdots\\
f_{N,i}
\end{array}
\right] =
\left[
\begin{array}{c}
-\frac{1-\beta^2}{(1+\sigma^2)\beta}\\
-\frac{1-\beta^2}{(1+\sigma^2)}\\
\vdots\\
-\frac{\beta^{N-2-i}(1-\beta^2)}{(1+\sigma^2)}
\end{array}
\right],
\]

\noindent which completely defines the structure of $\bF$.

Utilizing this structure of $\bF$, the Frobenius norm of $\bF$ can be computed to be
\begin{equation}
\left\|\bF\right\|_{F}^{2} = \frac{1}{(1+\sigma^2)^2}\left[\beta^{2(N-1)} + \frac{N-1}{\beta^2} - N\right]
\label{fro}
\end{equation}

\noindent Using this result and the bound $\left\|\bF\right\|_{F}^2 \leq (1+\sigma^2)^{-1}N\gamma\rho$, we find that the $\beta$ that meets the bound is the smallest positive root of
\begin{equation}
\beta^{2N}-(N+(1+\sigma^2)N\gamma\rho)\beta^2 + (N-1).
\label{betaeq}
\end{equation}
\end{proof}
\end{appendix}




\IEEEbiographynophoto{Zachary Chance (S'08)}
received the B.S. in electrical engineering at Purdue University, West Lafayette, Indiana in August 2007.  His research interests include adaptive communication systems, general feedback systems, and information theoretic radar imaging.  Mr. Chance is an active reviewer for IEEE Transactions on Wireless Communications and EURASIP Journal on Wireless Communications.

Mr. Chance has received the Ross Fellowship from Purdue University along with the Frederic R. Muller scholarship, Mary Bryan scholarship, and the Schlumberger-Tellkamp-Power scholarship.

\IEEEbiographynophoto{David J. Love (S'98 - M'05 - SM'09)}
received the B.S. (with highest honors), M.S.E., and Ph.D. degrees in electrical engineering from the University of Texas at Austin in 2000, 2002, and 2004, respectively. Since August 2004, he has been with the School of Electrical and Computer Engineering, Purdue University, West Lafayette, IN, where he is now an Associate Professor. Dr. Love currently serves as an Associate Editor for the IEEE Transactions on Signal Processing and the IEEE Transactions on Communications. He has also served as a guest editor for special issues of the IEEE Journal on Selected Areas in Communications and the EURASIP Journal on Wireless Communications and Networking. His research interests are in the design and analysis of communication systems, MIMO array processing, and array processing for medical imaging.

Dr. Love has been inducted into Tau Beta Pi and Eta Kappa Nu. Along with co-authors, he was awarded the 2009 IEEE Transactions on Vehicular Technology Jack Neubauer Memorial Award for the best systems paper published in the IEEE Transactions on Vehicular Technology in that year.  He was the recipient of the Fall 2010 Purdue HKN Outstanding Teacher Award.  In 2003, Dr. Love was awarded the IEEE Vehicular Technology Society Daniel Noble Fellowship.

\end{document}